\documentclass[reqno]{amsart}

\usepackage[linesnumbered,ruled,norelsize,vlined,algo2e]{algorithm2e}
\SetKwInput{KwInput}{Input}                % Set the Input
\SetKwInput{KwOutput}{Output}              % set the Output

\usepackage{lineno,hyperref}

\usepackage{booktabs} % For formal tables
\usepackage{amssymb,latexsym,amsfonts,amsmath,amsthm}
\usepackage{graphicx}
\usepackage{paralist}
\usepackage{xcolor}
\usepackage{dsfont}
\usepackage[export]{adjustbox}
\usepackage{tcolorbox}
\usepackage{chngcntr}
\usepackage[normalem]{ulem}
\usepackage{subcaption}
\usepackage{enumitem}
\usepackage{eso-pic}
\usepackage{epsfig}
\usepackage{mathtools}
\usepackage{epstopdf}
\usepackage{graphicx}
\usepackage{epsfig}
\usepackage{mathtools}
\usepackage{lineno}

\usepackage[justification=justified]{caption}
\allowdisplaybreaks

\newcommand{\tup}{\textup}

\newtheorem{theorem}{Theorem}
\newtheorem{definition}[theorem]{Definition}

\newtheorem{assumption}[theorem]{Assumption}
\newtheorem{proposition}[theorem]{Proposition}
\newtheorem{remark}[theorem]{Remark}

\DeclareMathOperator{\id}{id}
\usepackage{color}

\usepackage{xspace}	
\newcommand{\trn}{^{\scriptscriptstyle \top}}%
\newcommand{\N}{\mathbb{N}}%
\newcommand{\R}{\mathbb{R}}%

\def\K{\mathcal{K}}%
\def\Kinf{\mathcal{K}_\infty}%
\newcommand{\Rp}{\R_{\geq 0}}%
\let\ul=\underline%

\newcommand{\nn}{\mathbb{N}_0}%
\newcommand{\ba}{\ell^{\infty}}

\newcommand{\Let}{:=}
\newcommand{\I}{\mathrm{id}}
\newcommand{\V}{\mathcal{V}}
\newcommand{\n}{\mathcal{N}}
\newcommand{\m}{\mathcal{M}}

\makeatletter
\long\def\@maketablecaption#1#2{\@tablecaptionsize
	\global \@minipagefalse
	\hbox to \hsize{\parbox[t]{\hsize}{\centering #1 \\ #2}}}% ADDED
\makeatother
\begin{document}
	
\begin{abstract}
This paper presents a compositional framework for the construction of symbolic models for a network composed of a countably infinite number of finite-dimensional discrete-time control subsystems.
We refer to such a network as infinite network.
The proposed approach is based on the notion of alternating simulation functions.
This notion relates a concrete network to its symbolic model with guaranteed mismatch bounds between their output behaviors.
We propose a compositional approach to construct a symbolic model for an infinite network, together with an alternating simulation function, by composing symbolic models and alternating simulation functions constructed for subsystems. 
Assuming that each subsystem is incrementally input-to-state stable and  under some small-gain type conditions, we present an algorithm for orderly constructing local symbolic models with properly designed quantization parameters. In this way, the proposed compositional approach can provide us a guideline for constructing an overall symbolic model with any desired approximation accuracy.
A compositional controller synthesis scheme is also provided to enforce safety properties on the infinite network in a decentralized fashion. The effectiveness of our result is illustrated through a road traffic network consisting of infinitely many road cells.   
\end{abstract}
	
	\title[Symbolic Models for Infinite Networks of Control Systems: A Compositional Approach]{Symbolic Models for Infinite Networks of Control Systems: A Compositional Approach}

	\author{Siyuan Liu$^1$}
	\author{Navid Noroozi$^2$}
	\author{Majid Zamani$^{3,2}$}
	\address{$^1$Electrical and Computer Engineering Department, Technical University of Munich, Germany.}
	\email{sy.liu@tum.de}
	\address{$^2$Computer Science Department, Ludwig Maximilian University of Munich, Germany.}
	\email{navid.noroozi@lmu.de}
	\address{$^3$Computer Science Department, University of Colorado Boulder, USA.}
	\email{majid.zamani@colorado.edu}
	\maketitle
	
\section{Introduction}
Over the past few decades, large-scale interconnected systems have emerged in a wide range of safety-critical applications, such as traffic networks, smart manufacturing, and power networks.
In such applications, the number of agents can be extremely large, possibly unknown, or even vary over time as agents plug in and out. 
Unless rigorously addressed, such scalability issues may dramatically degrade system performance~\cite{infplatoon,Bamiehtac}.
%Motivated by the scalability issues arise in the aforementioned scenarios, 
It is a reasonable strategy to over-approximate the original network with the limit case in \emph{size}, in the sense that we introduce a network having infinitely many subsystems which includes the original network.
We call  this over-approximated network  an \emph{infinite} network  \cite{kawan2019lyapunov,mironchenko2020nonlinear,dashkovskiy2019stability}.
Note that as a special class of infinite-dimensional systems, infinite networks require a  rigorous treatment with careful choice of the infinite-dimensional state space of the overall network.  
%treat the large-scale system as a network composed of (countably) infinite number of finite-dimensional subsystems interacting with each other, which is usually called \emph{infinite network} in the literature
%Infinite networks are naturally seen as over-approximations of finite but very large networks, whereas finite networks, on the other hand, can be obtained by truncations of infinite networks. 
It is widely acknowledged that infinite networks capture the essence of the original network, in the sense that functionality indices, e.g. transient and steady-states behaviors, of an infinite network are preserved for its corresponding original finite network; see, e.g.,~\cite{infplatoon,navidcdc}.
In that way, one can eventually develop \emph{scale-free} (i.e., independent of the system size) approaches for the analysis and control of a finite, but arbitrarily large network~\cite{mironchenko2020input,dashkovskiy2020stability,navidcdc,SwikirIfac}.  

This paper is mainly concerned with symbolic controller synthesis  for infinite networks.
In the past few years, symbolic model (a.k.a. finite abstraction) based techniques have been widely developed to assist in the analysis/synthesis of controllers enforcing complex specifications which are difficult to handle using classical control design methods \cite{Tabuada.2009,Pola.2009,Pola.2008}. 
Specifically, symbolic models are abstract descriptions of original dynamics. In this regard, one can first build up a symbolic model of 
the original complex system, then perform analysis or synthesis over the symbolic model in an automated fashion (employing automata-theoretic techniques developed in the computer science literatures \cite{baier2008principles}), and finally translate the results back to the original system with correctness guarantees. 
A major challenge in the construction of symbolic models for large-scale networks is the \emph{curse of dimensionality}, i.e., the computational complexity of constructing symbolic models grows exponentially with the dimension of the system.
%scales up to only a few variables but run out of time or memory when confronted with large-scale models. 
In this paper, we aim at proposing a scale-free approach 
to alleviate the computational complexity in the construction of symbolic models for arbitrarily large-scale (potentially infinite) networks. A promising solution is to apply a \emph{divide and conquer} scheme, namely, \emph{compositional approach}.
In this framework, the overall network is decomposed into a set of finite lower-dimensional subsystems, for which symbolic models can be individually constructed in a computationally efficient way.
Then, a symbolic model for the overall network can be obtained by aggregating those of the subsystems.
Various compositional approaches have been explored in the past decade for the construction of symbolic models; see, e.g., \cite{tazaki2008bisimilar,Pola.2016,Majid18,SWIKIR2019,mallik2018compositional,meyer,kim2018constructing}. The results in \cite{tazaki2008bisimilar,Pola.2016,SWIKIR2019,mallik2018compositional} leverage small-gain type conditions to compositionally construct so-called \emph{complete} abstractions for a finite network.
The results presented in \cite{Majid18} introduce a different compositionality framework based on dissipativity theory.
The recent results in \cite{meyer,kim2018constructing} provide compositional construction of so-called \emph{sound} abstractions for a large-scale system without imposing compositionality conditions.
Although promising, all of the above-mentioned compositional approaches are typically tailored to a network composed of a \emph{finite} number of subsystems and do not address the scalability issues discussed earlier.
%can not be directly applied to infinite networks. 

In this paper, we develop a compositional approach for the construction of symbolic models for \emph{infinite} networks. We first introduce a notion of so-called alternating simulation functions used to relate an infinite network to its symbolic model with bounded mismatch between the output behaviors of them. 
Then, we provide a compositionality result showing that an overall symbolic model can be obtained by composing those of subsystems.
%Symbolic models for individual systems are constructed separately and efficiently. 
%and by exploiting the interconnection topology of 
Particularly, for a network composed of infinitely many incrementally input-to-state stable control subsystems, 
we leverage a recently presented small-gain theorem \cite{mironchenko2020nonlinear} and present an algorithm to design quantization parameters for the construction of local symbolic models and local simulation functions in a systematic way.
In particular, we give a \emph{top-down}, still compositional, algorithm computing local quantization parameters with the guarantee of obtaining an overall symbolic model with \emph{any} desired precision.
This differentiates our approach from existing ones as discussed in the sequel (cf. Related Works below).
Moreover, we present a decentralized controller synthesis approach for an infinite network that needs to meet \emph{safety} specifications.
It is shown that by composing local safety controllers which are synthesized for subsystems individually, the resulting overall controller enforces the safety specification on the overall infinite network with a formal guarantee.
Finally, the effectiveness of our proposed framework is verified through a road traffic network containing infinitely many road cells.

\textbf{Related Works.} There have been several attempts on the construction of symbolic models for infinite-dimensional systems~\cite{pola2010symbolic,girard2014approximately}.
The result in~\cite{pola2010symbolic} deals with continuous time-delay systems, for which symbolic models are obtained by projecting the infinite-dimensional functional state-space on a finite-dimensional subspace. The result in~\cite{girard2014approximately} provides a state-space discretization-free approach which can be applied to possibly infinite-dimensional incrementally stable control systems. 
%the results in~\cite{pola2010symbolic,girard2014approximately} are developed for a single (time-delay) infinite-dimensional system.
Although the results in~\cite{pola2010symbolic,girard2014approximately} are developed for (time-delay) infinite-dimensional systems, either state or input sets can be infinite dimensional. Here, we allow both state and input sets to be in infinite-dimensional space.
Moreover, the results in~\cite{pola2010symbolic,girard2014approximately} take a monolithic view of the systems while constructing symbolic models.
Therefore, in the case of potential application to an infinite network, the results in~\cite{pola2010symbolic,girard2014approximately} lose the network structure, and hence, they cannot be used for distributed control purposes.
Here, we propose a compositional approach for the construction of symbolic models for infinite networks, such that the network structures are preserved. 
A preliminary investigation of our proposed method appeared in~\cite{SwikirIfac}.
Our results here improve and extend those in~\cite{SwikirIfac} in three directions:
1) In this paper, we provide a detailed and mature description of the results presented in \cite{SwikirIfac}, including all proofs.
2) Here, we provide a \emph{top-down} compositional framework: under certain small-gain type conditions, an algorithm is provided as a guideline to orderly design local quantization parameters with the guarantee of obtaining an overall symbolic model with any desired precision. 
Whereas~\cite{SwikirIfac} presents a \emph{bottom-top} design approach, in the sense that one needs to first design local quantization parameters for subsystems and then use them to compute the overall approximation error.
%More precisely, the algorithm in~\cite{SwikirIfac} does not ensure achievement of a desired approximation error for the overall network, while here we do provide such a formal guarantee.
%In~\cite{SwikirIfac}, a desired approximation error might be met by trial and error: one is required to reiterate the algorithm by modifying the local parameters until the desired overall error is achieved (cf. \cite[Remark 4.5]{SWIKIR2019}).
%Since the alternating simulation function in~\cite{SwikirIfac} is defined with a max-form decay condition, the approach used in~\cite[Theorem 4.4]{SwikirIfac} (cf. similar arguments in~\cite[Theorem 4.3]{SWIKIR2019}) to transfer an additive form of Lyapunov function to a max form simulation function leads to inevitable conservatism, in the sense of resulting in a larger approximation error.
3) In comparison with the proposed results in~\cite{SwikirIfac}, due to the less conservatism in the definitions of alternating simulation functions in our work (cf. Definitions \ref{SF} and \ref{def:SFD1}), our compositional approach can potentially provide symbolic models for infinite networks with much smaller approximation errors (cf. case study in Section \ref{casestudy}).

\section{Preliminaries}\label{1:II}
\emph{Notation:} We denote by $\mathbb R$, $\nn$, and $\N$ the sets of real numbers, non-negative integers, and positive integers, respectively. We denote the closed, open, and half-open intervals in $\R$ by $[a,b]$, $(a,b)$, $[a,b)$, and $(a,b]$, respectively. For $a,b\in\nn$ and $a\le b$, we use $[a;b]$, $(a;b)$, $[a;b)$, and $(a;b]$ to denote the corresponding intervals in $\nn$. 
%Given any $a\in\R$, $\vert a\vert$ denotes the absolute value of $a$. 
Given any $\nu=(\nu_1,\cdots,\nu_n)\in\R^{n}$, we define by 
$|\nu|=\max_{1\leq i\leq n}|\nu_i|$ the infinity norm of $\nu$.  
%Elements of $\R^n$ are by default regarded as column vectors and we write $\nu\trn$ for the transpose of a vector $\nu\in\R^n$. 
%Given a symmetric matrix $A$, $\lambda_{\max}(A)$, and $\lambda_{\min}(A)$ denote the maximum and minimum eigenvalues of $A$, respectively.
%By $\ba$ we denote the Banach space of all infinite uniformly bounded sequences $s:=(s_i)\in\ba, i\in\N$, where $s_i$ denotes the $i$th position of a sequence $s\in\ba$. 
%Moreover, $\ba_+$ denotes the positive cone in $\ba$ consisting of all vectors $s\in\ba$ with $s_i\ge0,i\in\N$. For all $s,s'\in\ba$ we say that $s\leq s'$ if $s_i\leq s'_i$ for all $i\in\N$, and that $s\not\ge s$ if there is $i\in\N$ such that $s_i<s'_i$. 
%The standard unit vectors in $\ba$ are denoted by $e_i$, $i\in\N$; i.e., $e_i$ is the sequence of zeros with exception of position $i$, where the entry is $1$. 
%Given an operator $\Gamma:\ba_+ \rightarrow \ba_+$, \bb{$k\geq1\in\N$}, we define $\Gamma^{k}(\cdot):=\Gamma^{k-1}\circ\Gamma(\cdot)$, where $\Gamma^{0}$ is the identity	operator on $\ba$. 
We denote by $\text{card}(\cdot)$ the cardinality of a given set and by $\varnothing$ the empty set. 
For any set \mbox{$S\subseteq\R^n$} of the form of finite union of boxes, e.g., $S=\bigcup_{j=1}^MS_j$ for some finite number $M\in\N$, where $S_j=\prod_{i=1}^{n} [c_i^j,d_i^j]\subseteq \R^{n}$ with $c^j_i<d^j_i$, we define $\emph{span}(S)=\min_{j=1,\ldots,M}\eta_{S_j}$ and $\eta_{S_j}=\min\{|d_1^j-c_1^j|,\ldots,|d_{n}^j-c_{n}^j|\}$. 
Moreover, for a set in the form of $X= \prod_{i=1}^N X_i$, where $X_i \subseteq \R^{n_i}$, $\forall i\in[1;N]$, are of the form of finite union of boxes, and any positive (component-wise) vector $\phi = [\phi_1;\dots;\phi_N]$ with $\phi_i \leq \emph{span}(X_i)$, $\forall i\in [1;N]$, we define $[X]_\phi= \prod_{i=1}^N [X_i]_{\phi_i}$, where $[X_i]_{\phi_i} = [\R^{n_i}]_{\phi_i}\cap{X_i}$ and  $[\R^{n_i}]_{\phi_i}=\{a\in \R^{n_i}\mid a_{j}=k_{j}\phi_i,k_{j}\in\mathbb{Z},j=1,\ldots,n_i\}$.
Note that if $\phi = [\eta;\dots;\eta]$, where $0\leq\eta\leq\emph{span}(S)$, we simply use notation $[S]_{\eta}$ rather than $[S]_{\phi}$. 
%With a slight abuse of notation, we write $[S]_{0}:=S$. 
Note that $[S]_{\eta}\neq\emptyset$ for any $0\leq\eta\leq\emph{span}(S)$.
We use the notations $\mathcal{K}$ and $\mathcal{K}_\infty$ to denote different classes of comparison functions, as follows: $\mathcal{K}=\{\alpha:\mathbb{R}_{\geq 0} \rightarrow \mathbb{R}_{\geq 0} |$ $ \alpha$ is continuous, strictly increasing, and $\alpha(0)=0\}$; $\mathcal{K}_\infty=\{\alpha \in \mathcal{K} |$ $ \lim\limits_{r \rightarrow \infty} \alpha(r)=\infty\}$. For $\alpha,\gamma \!\in\! \mathcal{K}_{\infty}$ we write $\alpha\!\le\!\gamma$ if $\alpha(r)\!\le\!\gamma(r)$, and, with a slight abuse of the notation, $\alpha\!=\!c$ if $\alpha(r)\!=\!cr$ for all $c,r\!\geq\!0$.
%For $\alpha,\gamma \in \mathcal{K}_{\infty}$ we write $\alpha\le\gamma$ if $\alpha(r)\le\gamma(r)$ and $\alpha=0$ if $\alpha(r)=0$ for all $r\geq0$.
Finally, we denote by $\I$ the identity function over $\Rp$, i.e., $\I(r)=r$ for all $r\in \Rp$.

\subsection{Infinite networks} \label{infdef}
In this paper, we study the interconnection of a countably infinite number of discrete-time control subsystems.
Using $\N$ as the index set, the $i$-th subsystem	 
%\begin{definition} \label{subsys} 
%	A discrete-time control subsystem  $\Sigma{_i}, i \in \N,$ 
is denoted by a tuple $\Sigma{_i} = (X{_i},U{_i},W{_i},f{_i},Y{_i},h{_i})$, where $X{_i}\subseteq\R^{n_i}$, $U{_i}\subseteq\R^{m_i}$, $W{_i}\subseteq\R^{p_i}$, $Y{_i}\subseteq\R^{q_i}$, are the state, external input, internal input, and output set, respectively.
%Sets $\mathcal{U}{_i}$ and $\mathcal{W}{_i}$ denote the set of all bounded measurable functions $\nu:\N\rightarrow  U{_i}$ and $\omega:\N\rightarrow W{_i}$, respectively. 
The set valued map $ f{_i}: X{_i} \times U{_i} \times W{_i} \rightrightarrows  X{_i}$ is the state transition function and $h{_i} : X{_i} \rightarrow Y{_i} $ is the output map.	
The discrete-time control subsystem $\Sigma{_i}$ is described by difference inclusions of the form
\begin{align}\label{eq:2}
	\Sigma_i:\left\{
	\begin{array}{rl}
		{\mathbf{x}}{_i}(k+1)\in& f{_i}(\mathbf{x}{_i}(k),\nu{_i}(k),\omega{_i}(k)),\\
		\mathbf{y}{_i}(k)=&h{_i}(\mathbf{x}{_i}(k)),
	\end{array}
	\right.
\end{align}
where $\mathbf{x}_i:\nn\rightarrow  X_i $, $\mathbf{y}_i:\nn\rightarrow  Y_i$,  $\nu_i:\nn\rightarrow U_i$, and $\omega_i:\nn\rightarrow W_i$ are the state, output,  external input, and internal input signals, respectively.
%\end{definition}  
System $\Sigma_i$ is called \emph{deterministic} if $\text{card}(f_i(x_i,u_i,w_i))\leq1$, $\forall x_i\in  X_i, \forall  u_i\in  U_i, \forall  w_i \in  W_i$, and \emph{non-deterministic} otherwise. System $\Sigma_i$ is called \emph{discrete} if $ X_i, U_i, W_i$ are finite sets, and \emph{continuous} otherwise. 

Throughout the paper, we assume that each subsystem $\Sigma_i$ is affected by \emph{finitely} many neighbors.
For each $i\in \N$, the set of \emph{in-neighbors} of $\Sigma_i$ is denoted by $\n_i \subset \N \setminus \{i\}$, i.e. the set of subsystems $\Sigma_j$, $j\in\n_i$, directly  influencing $\Sigma_i$.
On the other hand, the set of \emph{out-neighbors} of $\Sigma_i$, denoted by $\m_i \subset \N  \setminus \{i\}$, is the set of $\Sigma_j$, $j\in\m_i$, directly affected by $\Sigma_i$.   
Sets $\n_i$ and $\m_i$ are finite, though not necessarily uniformly.
%let $\n_i$ and $\m_i$ be finite subsets of $\N$. Here, the index sets $\n_i$ and $\m_i$ enumerate the neighbors of $\Sigma_i$, i.e., those systems $\Sigma_j,j \in \n_i$, $\Sigma_{j'},j' \in \m_i$ that affect or are affected by $\Sigma_i$, respectively. 
%By definition we require that $i\notin \n_i\cup\m_i$, $\forall i\in \N$. Since $\n_i$ and $\m_i$ are finite subsets of $\N$, each $\Sigma_i$ can have only a finite number of neighbors.	
Formally, the input-output structure of each subsystem $\Sigma_i$, $i\in \N$, is given by
\begin{align}	
	\label{internalinput}	
	w_i &= (w_{ij})_{j\in\n_i}\in  W_i:=\prod_{j\in{\n_i}}W_{ij},\\
	\label{output}	
	y_i &=(y_{ij})_{j\in(i\cup{\m_i})}\in Y_i:=\prod_{ j\in (i\cup{\m_i})} \!\!Y_{ij},\\
	\label{outputfunction}
	h_i(x_i) &=(h_{ij}(x_i))_{j\in(i\cup{\m_i})}, 
	%&\!\!\!h_i(x_i)\!\! =\!( h_{ii}(x_i);h_{i\min\{\m_i\}}(x_i);\dots;h_{i\max\{\m_i\}}(x_i)).
\end{align}
with $w_{ij} \in W_{ij}$, $y_{ij} =h_{ij}(x_i)\in Y_{ij} $.
%\begin{align}
%	\label{outputfunction}
%	h_i(x_i) = [h_{ii}(x_i);h_{i\min\{\m_i\}}(x_i);\dots;h_{i\max\{\m_i\}}(x_i)].
%\end{align}
The outputs $y_{ii}$ are considered as external ones, whereas $y_{ij}$, $j\in\m_i$, are interpreted as internal ones which are used to construct an interconnection of subsystems. 
%The dimension of $w_{ij}$ is assumed to be equal to that of $y_{ji}$ for all $i\in\N$ and for all $j\in\n_i$. 
%In the case that no connection exists between subsystems $\Sigma_i$ and $\Sigma_j$, we simply have $h_{ij} \equiv 0$. 

%Furthermore, if for all $x_i\in  X_i$ 
%there exist $ u_i\in  U_i$ and $ w_i \in  W_i $ such that $\text{card}(f_i(x_i,u_i,w_i))\neq0$ we say the system is non-blocking. In this paper,
%we assume that all subsystems are non-blocking.

In the sequel, we denote by $\ba$ the Banach space of all uniformly bounded sequences $s=(s_i)_{i\in\N} \in\ba$, where $s_i \in \R^{n_i}$ denotes the $i$-th position of a sequence $s\in\ba$. The $\ba$ space is defined as  
\begin{align}\label{infnorm}
	\ba(\N,(n_i)) \Let \left\{s=(s_i)_{i\in\N}: s_i \in \R^{n_i}, \sup_{i\in\N} |s_i| < \infty \right \},
\end{align}	
endowed with the norm  $\Vert s \Vert \Let \sup_{i\in\N} |s_i|$.
Moreover, we use $\ba_+$ to denote the positive cone in $\ba$ consisting of all vectors $s\in\ba$ with $s_i\ge0,i\in\N$. We denote by $\textup{int}(\ba_+)$ the interior of $\ba_+$.

Now, we are ready to provide a formal definition of the infinite network.
\begin{definition}
	\label{interconnectedsystem} 
	Consider subsystems $\Sigma_i = (X_i,U_i,W_i,f_i,Y_i,h_i),~i\in \N,$ with input-output structure given by \eqref{internalinput} to \eqref{outputfunction}.
	An infinite network is formally a tuple $\Sigma = (X,U,f,Y,h)$, where 
	$X = \{ x = (x_i)_{i\in\N} : x_i \in X_i\}$, $U = \{ u = (u_i)_{i\in\N} : u_i \in U_i\}$, $f(x,u) =\{(x^+_i)_{i\in\N}|x_i^+\in f_i(x_i,u_i,w_i)\}$, $Y =\prod_{i\in\N}Y_{ii}$, $h(x) = (h_{ii}(x_i))_{i\in\N}$.
	%   The infinite network is described by
	%	\begin{align*}
	%	\Sigma:\left\{
	%	\begin{array}{rl}
	%	{\mathbf{x}}(k+1)\in& f(\mathbf{x}(k),\nu(k)),\\
	%	\mathbf{y}(k)=&h(\mathbf{x}(k)).
	%	\end{array}
	%	\right.
	%	\end{align*}
	A \emph{concrete} infinite network $\Sigma = (X,U,f,Y,h)$, denoted by $\Sigma=\mathcal{I}(\Sigma_i)_{i\in\N}$, consists of infinitely many \emph{continuous} subsystems $\Sigma_i$, with the interconnection variables constrained by 
	\begin{align}\label{const}
		\forall i\in\N,\forall j\in\n_i,w_{ij}=y_{ji},Y_{ji}\subseteq W_{ij}.
	\end{align}
	An \emph{abstract} infinite network $\hat \Sigma = (\hat X,\hat U,\hat f,\hat Y,\hat h)$, denoted by $\hat \Sigma=\mathcal{I}(\hat \Sigma_i)_{i\in\N}$, is composed of infinitely many \emph{discrete} subsystems, with the interconnection variables constrained by
	\begin{align}\label{const1}
		\forall i\in\N,\forall j\in\n_i, | \hat y_{ji} - \hat w_{ij} |  \leq  \phi_{ij},[\hat {{Y}}_{ji}]_{\phi_{ij}}  \subseteq  \hat {{W}}_{ij},
	\end{align}	
	where $\phi_{ij} \in \R_{\geq 0}$ is an internal input quantization parameter designed later (cf. Definition \ref{def:sym}).
	%for constructing local symbolic models (cf. Definition \ref{def:sym}).
\end{definition}
{Throughout the paper, we assume that $f(x,u) \in X$ for all $(x,u) \in X \times U$, which ensures that the infinite network is well-posed.}

\begin{remark}
	Note that in Definition~\ref{interconnectedsystem}, the interconnection constraint in~\eqref{const} for the concrete network is different from~\eqref{const1} for the abstract network.
	For a network of symbolic models,  we allow for possibly different granularities of finite internal input sets $\hat {W}_{ij}$ and output sets $\hat {Y}_{ji}$, and introduce parameters $\phi_{ij}$ in \eqref{const1} for having a well-posed interconnection.
	The values of $\phi_{ij}$ will be designed later in Definition \ref{def:sym} while constructing local symbolic models of subsystems.
	\hfill$\diamond$ 
\end{remark}

%
%\begin{definition}
%	\label{absnetsw} 
%Consider abstract discrete-time control subsystems $$\hat \Sigma_i = (\hat X_i,\hat W_i,\hat U_i,\hat f_i,\hat Y_i,\hat h_i),~i\in \N,$$ with input-output structure given by \eqref{internalinput}-\eqref{outputfunction}.
%The abstract infinite network is then formally a tuple $\hat \Sigma = (\hat X,\hat U,\hat f,\hat Y,\hat h)$, where 
%\begin{align*}
%\hat X &= \{ \hat x = (\hat x_i)_{i\in\N} : \hat x_i \in \hat X_i,\ \Vert \hat x\Vert:=\sup_{i\in\N}\{|\hat x_i|\} < \infty \},\\
%\hat U &= \{ \hat u = (\hat u_i)_{i\in\N} : \hat u_i \in \hat U_i,\ \Vert \hat u\Vert:=\sup_{i\in\N}\{|\hat u_i|\}< \infty \},\\
%\hat Y&=\prod_{i\in\N}\hat Y_{ii},\\
%\hat f(x,u)&=\{(\hat x^+_i)_{i\in\N}|\hat x_i^+\in \hat f_i(\hat x_i,\hat w_i,\hat u_i)\},\\ 
%\hat h(x) &= (\hat h_{ii}(\hat x_i))_{i\in\N}.
%\end{align*} 
%The abstract infinite network is denoted by $\hat \Sigma=\mathcal{I}(\hat \Sigma_i)_{i\in\N}$
%%, and described by
%%\begin{align*}
%%\hat \Sigma:\left\{
%%\begin{array}{rl}
%%\hat{\mathbf{x}}(k+1)\in& \hat f(\hat {\mathbf{x}}(k),\hat \nu(k)),\\
%%\hat {\mathbf{y}}(k)=&\hat h(\hat{ \mathbf{x}}(k)).
%%\end{array}
%%\right.
%%\end{align*}
%Moreover, the interconnection variables are constrained by 
%\begin{align}\label{const}
%\forall i\in\N,\forall j\in\n_i,\Vert \hat y_{ji}\!-\!\hat w_{ij} \Vert  \leq  \phi_{ij},[\hat {{Y}}_{ji}]_{\phi_{ij}} \!\subseteq\! \hat {{W}}_{ij}.
%\end{align}	
%where $\phi_{ij}$ is an internal input quantization parameter designed for constructing local finite abstractions (cf. Subsection \ref{absCur}).
%\end{definition} 

\subsection{Alternating simulation functions} 
Here, we provide a notion of alternating simulation functions which quantitatively relate two infinite networks.
\begin{definition}\label{SF}
	Consider infinite networks $\Sigma\!=(X,U,f,Y,h)$ and $\hat{\Sigma}\!=(\hat{{X}},\hat{{U}},\hat{f},$ $\hat{{Y}},\hat{h})$, where $\hat{{Y}}\subseteq{{Y}}$. For $\varpi \in \mathbb R_{\ge 0}$, a function $\tilde V: X \times \hat{{X}} \rightarrow \mathbb R_{\ge 0}$ is called an $\varpi$-approximate alternating simulation function ($\varpi$-ASF) from $\hat{\Sigma}$ to $\Sigma$, if there exists a function $\alpha \in \mathcal{K_{\infty}}$ such that 
	\begin{enumerate}[leftmargin=*]
		\item[(i)] For all $x \in X$, $\hat x \in \hat {X}$, one has 
		\begin{align}
			\alpha(\Vert h(x) - \hat h(\hat x) \Vert) \leq \tilde V(x,\hat x);
		\end{align}
		\item[(ii)]  For all $x \in X$ and $\hat x \in \hat {X}$ with $\tilde V(x,\hat x) \leq \varpi $, for all $\hat u\in \hat{{U}}$, there exists $u \in U$ such that for all $x^+ \in f(x,u)$, there exists $\hat x^+ \in \hat{f}(\hat{x},\hat{u})$ so that 
		\begin{align}
			\tilde V(x^+,\hat x^+) \leq \varpi.
		\end{align}
		%		(b) $\forall \hat u\in \hat{\mathbb{U}}$, $\forall \hat x^+ \in \hat{f}(\hat{x},\hat{u})$, $\exists u\in U$, $\exists x^+ \in f(x,u)$,  s.t. $\tilde V(x^+,\hat x^+) \leq \varpi$.
	\end{enumerate}
	If there exists an alternating simulation function from $\hat{\Sigma}$ to $\Sigma$, $\hat{\Sigma}$ is called an \emph{abstraction} of $\Sigma$. Additionally, if $\hat{\Sigma}$ is discrete ($\hat{X}$ and $\hat{U}$ are finite sets), $\hat{\Sigma}$ is called a \emph{symbolic model} (or \emph{finite abstraction}) of the concrete network $\Sigma$.
\end{definition}
\begin{remark}\label{remarkaccuracy}
	Definition \ref{SF} implies that the relation $R\subseteq{X}\times \hat{{X}}$ defined by $R=\left\{(x,\hat{x})\in {X}\times \hat{{X}}|\tilde{V}(x,\hat{x})\leq \varpi \right\}$ is an $\hat \varepsilon$-approximate alternating simulation relation, defined in \cite{Pola.2009}, from $\hat \Sigma$ to $\Sigma$ with $\hat \varepsilon = \alpha^{-1}(\varpi)$. As shown in \cite{Pola.2009}, the existence of an $\varpi$-ASF enables us to design a controller for the abstract network $\hat \Sigma$, and refine the controller back to the concrete network $\Sigma$.
	\hfill$\diamond$ 
\end{remark}

\section{Compositional Construction of Symbolic Models}
In this section, we provide a method for compositional construction of an alternating simulation function between two infinite networks $\Sigma=\mathcal{I}(\Sigma_i)_{i\in\N}$ and $\hat{\Sigma}=\mathcal{I}(\hat{\Sigma}_i)_{i\in\N}$. Here, we assume that each pair of subsystems $\Sigma_i=(X_i, W_i, U_i,f_i, Y_i,h_i)$ and 
$\hat{\Sigma}_i\!=(\hat{{X}}_i,\hat{{W}}_i,\hat{{U}}_i,\hat{f}_i,\hat{{Y}}_i,\hat{h}_i)$ admit a local alternating simulation function as defined next.
\begin{definition}\label{def:SFD1}
	Consider subsystems
	$\Sigma_i=(X_i,U_i,W_i, f_i,Y_i,h_i)$ and 
	$\hat{\Sigma}_i=(\hat{{X}}_i,$ $\hat{{U}}_i, \hat{{W}}_i, \hat{f}_i,\hat{{Y}}_i,\hat{h}_i)$
	where $\hat{{W}}_i\subseteq{W_i}$ and $\hat{{Y}}_i\subseteq{Y_i}$.
	Given $\varpi_i \in \R_{\ge 0}$, a function $V_{i} : X_i \times \hat{X}_i \rightarrow \R_{\ge 0}$ is called a local $\varpi_i$-ASF from $\hat{\Sigma}_i$ to $\Sigma_i$,
	%$\varpi_i$-approximate initial-state opacity-preserving simulation function ($\varpi_i$-InitSOPSF) 
	if there exist a constant $\vartheta_i \in \mathbb R_{\ge 0}$, and  functions $\underline{\alpha}_i, \overline{\alpha}_i \in \mathcal{K_{\infty}}$ such that 
	\begin{enumerate}[leftmargin=*]
		\item[(i)] For all $x_i \in X_i$, all $\hat x_i \in \hat {X}_i$, one has
		\begin{align}\label{localineq1}
			\underline{\alpha}_i(| h_i(x_i) - \hat h_i(\hat x_i) |) \leq V_{i}(x_i,\hat x_i)\leq{\overline{\alpha}}_i  (|x_i -\hat{x}_i|).
		\end{align}
		\item[(ii)] For all $x_i \in X_i$, all $\hat x_i \in \hat {X}_i$ with $V_{i}(x_i,\hat x_i) \leq \varpi_i$, for all $w_i \in W_i$, all $ \hat w_i \in \hat{{W}}_i$ with $| w_i- \hat w_i | \leq \vartheta_i$, for all $ \hat u_i \in \hat{{U}}_i$, there exists $u_i \in U_i$ such that for all $x_{i}^+ \in f_i(x_i,u_i,w_i)$, there exists $\hat x_{i}^+ \in \hat{f}_i(\hat{x}_i,\hat{u}_i,\hat{w}_i)$ so that 
		\begin{align}
			V_{i}(x_{i}^+,\hat x_{i}^+) \leq \varpi_i.
		\end{align}
		%	 $\forall u_i \in \mathbb U_i$, $\forall x_{d_i} \in f_i(x_i,u_i,w_i)$, $\exists \hat u_i \in \hat{\mathbb{U}}_i$, $\exists\hat x_{d_i} \in \hat{f}_i(\hat{x}_i,\hat{u}_i,\hat{w}_i)$, s.t. $V_{i}(x_{d_i},\hat x_{d_i}) \leq \varpi_i$.
		%			\item[(b)] $\forall \hat u_i \in \hat{\mathbb{U}}_i$, $\forall \hat x_{d_i} \in \hat{f}_i(\hat{x}_i,\hat{u}_i,\hat{w}_i)$, $\exists u_i \in \mathbb U_i$, $\exists x_{d_i} \in f_i(x_i,u_i,w_i)$,  s.t. $V_{i}(x_{d_i},\hat x_{d_i}) \leq \varpi_i$.
		%		\end{enumerate}
	\end{enumerate}
	If there exists a local alternating simulation function from $\hat{\Sigma}_i$ to $\Sigma_i$, $\hat{\Sigma}_i$ is called an abstraction of $\Sigma_i$. Additionally, if $\hat{\Sigma}_i$ is discrete ($\hat{X}_i$, $\hat{U}_i$, and $\hat{W}_i$ are finite sets), $\hat{\Sigma}_i$ is called a \emph{symbolic model}  (or \emph{finite abstraction}) of the concrete subsystem $\Sigma_i$.
\end{definition}
%Note that local alternating simulation functions of subsystems are mainly for constructing alternating simulation functions for the overall infinite networks and they are not directly used for deducing any approximate alternating simulation relation.

The next theorem provides a compositional approach for the construction of an alternating simulation function between two infinite networks using the above-defined local alternating simulation functions. 

\begin{theorem}\label{thm:3}
	Consider an infinite network $\Sigma=\mathcal{I}(\Sigma_i)_{i\in\N}$.
	Assume that each $\Sigma_i$ and its abstraction $\hat{\Sigma}_i$ admit a local $\varpi_i$-ASF $V_i$ equipped with functions $\underline{\alpha}_i, \overline{\alpha}_i \in \mathcal{K_{\infty}}$ and constants $\varpi_i, \vartheta_i \in\R_{\ge0}$
	as in Definition \ref{def:SFD1}. Suppose that there exist $\underline{\alpha},\overline{\alpha} \in \Kinf$, and constants $\underline{\varpi}, \varpi\in\R_{\ge0}$ such that for each $i\in\N$ 
	\begin{align}\label{mainineq1}
		\underline{\alpha} \leq\underline{\alpha}_i\leq \overline{\alpha}_i\leq\overline{\alpha},\\ \label{mainineq2}
		\underline{\varpi} \leq {\varpi_i} \leq {\varpi}.
	\end{align}For each $i \in \N$ and $j \in \n_i$, let functions $\underline\alpha_j$, constants $\varpi_j, \vartheta_i$, and constants $\phi_{ij}$ as in \eqref{const1} satisfy the following inequality
	\begin{align} \label{compoquaninit}
		%		\alpha^{-1}_{j}(\varpi_j) \leq \vartheta_i,		\\
		{\underline\alpha}^{-1}_{j}(\varpi_j) 	+ \phi_{ij} \leq \vartheta_i .
	\end{align}
	%	where $\phi_{ij}$ is an internal input quantization parameter as in \eqref{const1} for constructing the symbolic models $\hat{\Sigma}_i$,
	Then, function 
	\begin{align}\label{defVinit}
		\tilde{V}&(x,\hat{x})\Let\sup\limits_{i\in\N}\{ \frac{{\varpi}}{\varpi_i} V_i(x_{i},\hat{x}_{i})\},
	\end{align}
	is well-defined and it is an $\varpi$-ASF from  $\hat{\Sigma}=\mathcal{I}(\hat{\Sigma}_i)_{i\in\N}$ to $\Sigma=\mathcal{I}(\Sigma_i)_{i\in\N}$.
\end{theorem}
\begin{proof}
	First we show that function $\tilde{V}$ constructed as in \eqref{defVinit} is well-defined. Note that for all $x\in X$ and for all $\hat{x}\in \hat{X}$ we have
	\begin{align*}\notag
		\tilde{V}(x,\hat{x})&\Let\sup\limits_{i\in\N}\{ \frac{\varpi}{\varpi_i} V_i(x_{i},\hat{x}_{i}) \}
		\stackrel{\eqref{localineq1}} \leq \varpi \sup\limits_{i\in\N}\{ \varpi^{-1}_i  \overline{\alpha}_i  (|x_i -\hat{x}_i|) \} \\
		&\leq \varpi \sup\limits_{i\in\N}\{\varpi^{-1}_i  \overline{\alpha}_i  (|x_i| +|\hat{x}_i|) \} 
		\stackrel{\eqref{mainineq2}}\leq\varpi \sup\limits_{i\in\N}\{ \underline{\varpi}^{-1} \overline{\alpha} (|x_i| +|\hat{x}_i|) \} \\
		&\leq \frac{\varpi}{\underline{\varpi}}  \overline{\alpha}\sup\limits_{i\in\N}\{ |x_i| +|\hat{x}_i| \} 
		\leq \frac{\varpi}{\underline{\varpi}}  \overline{\alpha}(\sup\limits_{i\in\N}\{ |x_i| \}+\sup\limits_{i\in\N}\{|\hat{x}_i| \})\\
		&\stackrel{\eqref{infnorm}}\leq \frac{\varpi}{\underline{\varpi}}  \overline{\alpha}(\Vert x\Vert+\Vert \hat{x}\Vert )<\infty.
	\end{align*} 
	Next, we show that there exists ${\alpha} \in \mathcal{K}_{\infty}$ such that condition (i) of Definition \ref{SF} holds. Consider any $x\in X$, $\hat{x} \in {\hat{X}}$, one gets
	\begin{align*}\notag
		\Vert h(x)-\hat{h}(\hat{x})\Vert &=\sup\limits_{i\in\N}\{| h_{ii}(x_i)-\hat{h}_{ii}(\hat{x}_i)|\} 
		\stackrel{\eqref{outputfunction}}\leq\sup\limits_{i\in\N}\{| h_{i}(x_i)-\hat{h}_{i}(\hat{x}_i)|\}\\
		&\stackrel{\eqref{localineq1}} \leq\sup\limits_{i\in\N}\{\underline{\alpha}^{-1}_i(V_i(x_{i},\hat{x}_{i}))\} = \sup\limits_{i\in\N}\{ \underline{\alpha}^{-1}_i (\varpi_{i}\varpi^{-1}_i V_i(x_{i},\hat{x}_{i}))\}\\
		&\stackrel{\eqref{mainineq1}\eqref{mainineq2}} \leq \underline\alpha^{-1}\sup\limits_{i\in\N}\{ \varpi \varpi^{-1}_{i} (V_i(x_{i},\hat{x}_{i}))\}
		\stackrel{\eqref{defVinit}}=  \underline\alpha^{-1}(\tilde{V}(x,\hat{x})).
	\end{align*}
	Hence, condition (i) holds with $\alpha\Let \underline\alpha$.  
	Next, we show that condition (ii) of Definition \ref{SF} is satisfied.
	Let us consider any 
	$ x = (x_i)_{i\in\N} \in  X$ and $\hat x= (\hat x_i)_{i\in\N} \in\hat{{X}}$ such that $\tilde{V}(x,\hat x) \leq \varpi$. It can be seen that from the construction of $\tilde{V}$ in \eqref{defVinit}, we have $V_i(x_{i},\hat{x}_{i}) \leq \varpi_i$, for each $i \in \N$. For each pair of subsystems $\Sigma_i$ and $\hat{\Sigma}_i$, the internal inputs satisfy the following inequality
	%		\begin{align*}
	%		\Vert w_i- \hat{w}_i\Vert =& \max\limits_{j \in \tup{Pre}_I(i)}\{\Vert w_{ij}- \hat{w}_{ij}\Vert \} 
	%		= \max\limits_{j \in \tup{Pre}_I(i)}\{\Vert y_{ji}-\hat{y}_{ji}\Vert \}
	%		\leq\max\limits_{j \in \tup{Pre}_I(i)}\{\Vert h_{j}(x_j)\!\!-\!\!\hat{h}_{j}(\hat{x}_j)\Vert \}\\
	%		\leq&\max\limits_{j \in \tup{Pre}_I(i)}\{\alpha^{-1}_{j}\circ V_j(x_{j},\hat{x}_{j})\}
	%		\leq\max\limits_{j \in \tup{Pre}_I(i)}\{\alpha^{-1}_{j}\circ \varpi_j\}.		
	%		\end{align*} 
	\begin{align*} 
		| w_i- \hat{w}_i | =& \max\limits_{j \in \n_i}\{| w_{ij}- \hat{w}_{ij} |\} \stackrel{\eqref{const}}= \max\limits_{j \in \n_i}\{| y_{ji}-\hat{y}_{ji}+\hat{y}_{ji}-\hat{w}_{ij}|\} \\ 	\stackrel{\eqref{const1}}\leq &\max\limits_{j \in \n_i}\{| y_{ji}-\hat{y}_{ji}| + \phi_{ij}\}  
		\leq \max\limits_{j \in \n_i}\{| h_{j}(x_j) - \hat{h}_{j}(\hat{x}_j)|  + \phi_{ij}\}\\
		\stackrel{\eqref{localineq1}}\leq&\max\limits_{j \in \n_i}\{{\underline\alpha}^{-1}_{j} V_j(x_{j},\hat{x}_{j})  + \phi_{ij}\}
		\leq\max\limits_{j \in \n_i}\{{\underline\alpha}^{-1}_{j}( \varpi_j)  + \phi_{ij}\}.		
	\end{align*}
	Using \eqref{compoquaninit}, one has $| w_i- \hat{w}_i| \leq \vartheta_i$ for each $i \in \N$. 	Therefore, by Definition \ref{def:SFD1} for each pair of subsystems $\Sigma_i$ and $\hat\Sigma_i$, one has for any $\hat{u}_i\in \hat{ U}_i$, there exists $u_i \in U_i$ such that for any $x_{i}^+ \in f_i(x_i,u_i,w_i)$, there exists $\hat{x}_{i}^+ \in \hat{f}_i(\hat x_i,\hat u_i,\hat w_i)$ such that  $V_i(x_{i}^+,\hat{x}_{i}^+) \leq \varpi_i$. 
	As a result, we get for any $\hat u=(\hat u_i)_{i\in\N}\in\hat{{U}}$, there exists $u = (u_i)_{i\in\N} \in U$, such that for any $x^+= (x_i^+)_{i\in\N} \in f(x,u)$, there exists $\hat{x}^+ = (\hat{x}_i^+)_{i\in\N}\in \hat{f}(\hat{x},\hat{u})$ such that $\tilde{V}(x^+, \hat{x}^+) = \sup\limits_{i\in\N}\{ \frac{\varpi}{\varpi_i} V_i(x_{i},\hat{x}_{i}) \} \leq \varpi$. 
	Therefore, condition (ii) of Definition \ref{SF} is satisfied with $\varpi = \sup\limits_{i\in\N} \varpi_i$. Therefore, we conclude that $\tilde{V}$ is an $\varpi$-ASF from  $\hat{\Sigma}=\mathcal{I}(\hat{\Sigma}_i)_{i\in\N}$ to $\Sigma=\mathcal{I}(\Sigma_i)_{i\in\N}$.
\end{proof}
\begin{remark}
	{Note that practically speaking, the computation of a symbolic model consisting of infinite subsystems requires an infinite memory usage, which prevents us from having a central entity to handle the construction of a symbolic model for the overall network.}
	However, the proposed compositional framework is still needed to formally establish the alternating simulation relation between infinite networks in terms of preserving desired properties.	On this basis, one can develop decentralized (or distributed) schemes to solve controller synthesis problems compositionally using symbolic models of subsystems. 
	\hfill$\diamond$ 
\end{remark}

Next we provide a method to construct local symbolic models together with corresponding local alternating simulation functions for the concrete subsystems under \emph{incremental} stability-type conditions.  
%will impose conditions on the dynamics of the concrete subsystems such that one can 
\subsection{Construction of local symbolic models}\label{1:IV}
In this subsection, we present a method to construct a symbolic model $\hat{\Sigma}_i$, together with the corresponding local alternating simulation function, for a given finite-dimensional deterministic subsystem $\Sigma_i$.
Consider a subsystem $\Sigma_i=(X_i, U_i, W_i, f_i, Y_i,h_i)$ as in \eqref{eq:2}.
{Assume that there exists $\ell\in\K$ such that the output map $h_i$ satisfies $| h_i(x_i) - h_i(x_i')|\leq \ell(|x_i-x_i'|)$ for all $x_i, x_i'\in  X_i$.}
%\NN{This is not a Lipschitz condition, but a so-called $\K$-boundedness. BTW, does $\ell$ really need to be of class $\Kinf$, see my TAC'18 paper with Lars Grüne, et al.? What if it simply lies in $\K$?}. Additionally, let $\Sigma_i$ be incrementally input-to-state stable ($\delta$-ISS)~\cite{angeli} as given next.
Additionally, let $\Sigma_i$ be incrementally input-to-state stable ($\delta$-ISS) \cite{angeli2002lyapunov} as defined next.

\begin{definition}\label{iss} 
	System $\Sigma_i$ is incrementally input-to-state stable ($\delta$-ISS) if there exist a so-called $\delta$-ISS Lyapunov function $ \V_i: X_i\times  X_i \to \mathbb{R}_{\geq0} $, and functions ${\underline{\psi}_i}, {\overline{\psi}_i}, \kappa_i,\rho_{w_i} , \rho_{u_i} \in \mathcal{K}_{\infty}$, with $\kappa_i<\I$ such that for all $x_i,x'_i\in X_i$,  all $w_i, w'_i\in W_i$, and all $u_i, u'_i\in U_i$  
	\begin{align}\label{c1}
		\underline{\psi}_i (|x_i-x'_i| ) \leq &\V_i(x_i,x'_i)\leq \overline{\psi}_i (| x_i-x'_i| ),\\ \label{c2}
		\V_i(f_i( x_i,u_i,w_i),f_i( x'_i,u'_i,w'_i))\leq & \kappa_i (\V_i(x_i,{x_i'}))+\varrho_{w_i}(| w_i-w'_i|)+\varrho_{u_i}(| u_i-u'_i| ).
	\end{align}We further assume that there exists $\hat \gamma_i \in \Kinf$ such that for all $x_i,x_i',x_i'' \in {X}_i$
	\begin{align}\label{tinq} 
		\V_i(x_i,x_i')\leq \V_i(x_i,x_i'')+ \hat\gamma_i(| x_i'-x_i''|).
	\end{align}
\end{definition}
Note that a typical $\delta$-ISS Lyapunov function  \cite{angeli2002lyapunov} does not require condition \eqref{tinq}. However, in most real-world applications, the state set ${X}_i$ of a concrete subsystem is restricted to a compact subset of $\mathbb{R}^n$, and hence, condition~\eqref{tinq} is \emph{not} restrictive \cite{zamani2014symbolic}.

%We denote by $\Sigma^{\delta}$ a switched system $\Sigma$ that admits common or multiple $\delta$-ISS Lyapunov functions.  
Now, we construct a symbolic model $\hat\Sigma_i$ of a $\delta$-ISS subsystem $\Sigma_i$ as follows.
\begin{definition}\label{def:sym}
	Let $\Sigma_i=(X_i, U_i, W_i,f_i, Y_i,h_i)$ be  $\delta$-ISS, where $X_i, U_i, W_i$ are assumed to be finite unions of boxes. Consider a symbolic model $\hat{\Sigma}_i=(\hat{{X}}_i,\hat{{U}}_i,\hat{{W}}_i,$ $\hat{f}_i,\hat{{Y}}_i,\hat{h}_i)$ with a tuple of parameters $q_i = (\eta^x_i,\eta^u_i,\phi_i)$, where:
	\begin{enumerate}[label={$\bullet$},leftmargin=*]
		\item ${\hat{X}_i}=[X_i]_{\eta^x_i}$, where $0\leq \eta^x_i\leq\emph{span}( X_i)$ is the state set quantization parameter; 
		\item ${\hat{U}_i}=[ U_i]_{\eta^u_i}$, where $0\leq \eta^u_i \leq\emph{span}( U_i)$ is the external input set quantization parameter;
		\item $\hat{W}_i =[{W}_i]_{{\phi}_i}$,  where $\phi_i$, satisfying $0 \leq | \phi_i | \leq  \emph{span}({W}_i)$, is the internal input set quantization parameter;
		%\item ${\hat{W}}=[{W}]_{\hat{\varpi}}$, where $0\leq\hat{\varpi}\leq\emph{span}(\mathbb W)$ is the internal input set quantization parameter;
		\item $\hat{x}^+_i \in \hat{f}_i(\hat{x}_i,\hat{u}_i,\hat{w}_i)$ if and only if $|\hat{x}^+_i - f_i(\hat{x}_i,\hat{u}_i,\hat{w}_i)|  \leq \eta^x_i$;
		\item $\hat{Y}_i=\{h_i(\hat{x}_i)\,\,|\,\,\hat{x}_i \in \hat{X}_i\} $;
		\item $\hat{h}_i=h_i$.
		%		\item $\hat{W}_i$ is an appropriate finite internal input set satisfying
		%		$\hat{W}_i\!=\!\prod_{j\in{\n_i}} \!\!\hat{W}_{ij}$ and $\hat{Y}_{ji}\subseteq \hat{W}_{ij}$ $\forall i\in\N,\forall j\in\n_i$.
	\end{enumerate}
\end{definition}

Now we are ready to establish a local alternating simulation relation between a $\delta$-ISS subsystem $\Sigma_i$ and its symbolic model $\hat \Sigma_i$ constructed as in Definition~\ref{def:sym} with suitably chosen quantization parameters.  
\begin{theorem}\label{thm:2}
	Let $\Sigma_i$ be $\delta$-ISS with the corresponding $\delta$-ISS Lyapunov function $\V_i$ satisfying~\eqref{c1} to~\eqref{tinq}  with functions   
	${\underline{\psi}_i}, {\overline{\psi}_i}, \kappa_i,\rho_{w_i}, \rho_{u_i}, \hat{\gamma}_i \in \Kinf$. 	For design parameters $\varpi_i$ and $\vartheta_i$, let $\hat{\Sigma}_i$ be a symbolic model constructed as in Definition~\ref{def:sym} with the quantization parameters  $\eta^x_i$ and  $\eta^u_i$ satisfying 
	\begin{align} \label{secquantinit}	
		%	\eta^x_i \leq \min\{\hat{\gamma}_i^{-1}[(1-\kappa_i) (\varpi_i)-\rho_{w_i}(\vartheta_i)-\rho_{u_i}(\eta^u_i)],  \overline{\psi}_i^{-1}(\varpi_i)\}.
		\eta^x_i \leq  {\hat{\gamma}_i^{-1}[(\I-\kappa_i) (\varpi_i)-\rho_{w_i}(\vartheta_i)-\rho_{u_i}(\eta^u_i)]}.
	\end{align}	  
	%\NN{I'm struggling with understanding the RHS. What does $(1-\kappa_i) (\varpi_i)$ mean, for instance?}
	Then, $\V_i$ is {a local $\varpi_i$-ASF both from $\hat{\Sigma}_i$ to $\Sigma_i$ and from $\Sigma_i$ to $\hat{\Sigma}_i$}. %\NN{I'm not following. Yre you saying that the function is establishing some sort of ``bisimilar'' relation between systems.}
	%	
	%		In particular, $\V_i$ satisfies \eqref{e:SFC11} and \eqref{e:SFC22} with $\underline{\alpha}_i=(\ell_i\circ\underline{\psi}^{-1}_i)^{-1}$, $\overline{\alpha}_i=\overline{\psi}_i$, $\sigma_i=\I-(\I-\varphi_i)(\I-\kappa_i)$, $\rho_{w_i}=(\I+\lambda_i)\circ{\kappa}^{-1}\circ\varphi^{-1}_i\circ\chi_i\circ\varrho_{wi}$, $\rho_{u_i}=0$, $\varepsilon_i=(\I+\lambda^{-1}_i)\circ\kappa^{-1}\circ\varphi^{-1}_i\circ\chi_i\circ(\chi_i-\I)^{-1}\circ{\gamma}_i(\eta^x_i)$, where $\lambda_i,\chi_i,\varphi_i$ are some arbitrarily chosen $\mathcal{K}_{\infty}$ functions with $\varphi_i<\I,~\chi_i>\I$. 
\end{theorem}
\begin{proof}
	First, we show that condition (i) in Definition \ref{def:SFD1} holds.
	%	Since $\Sigma$ is $\delta$-ISS as in \eqref{eq:ISTFC1},  and given the Lipschitz assumption on $h$, 
	Given the Lipschitz assumption on $h_i$ and by \eqref{c1}, for all $x_i \in {X}_i$ and $\hat{x}_i \in {\hat{X}_i}
	$, one gets the left inequality of \eqref{localineq1} as
	\begin{align}\notag
		|h_i(x_i)-\hat{h}_i(\hat{x}_i)|\leq \ell(| x_i-\hat{x}_i|)\leq\ell\circ\underline{\psi}_i^{-1}(\mathcal{V}(x_i,\hat{x}_i)),
	\end{align}	
	and the right inequality \eqref{localineq1} holds with $
	\mathcal{V}(x_i,\hat{x}_i)  \leq \overline{\psi}_i (| x_i-\hat{x}_i|)$.
	Hence, condition (i) in Definition \ref{def:SFD1} holds with $\underline{\alpha}_i=\ul\psi_i\circ\ell^{-1}$ and $\overline{\alpha}_i=\overline{\psi}_i$. 
	Now we show condition (ii) in Definition \ref{def:SFD1}. From \eqref{tinq}, {for all $x_i\in {X}_i, \hat{x}_i \in {\hat{X}_i}$, for all $u_i\in{{U}_i}, \hat{u}_i \in {\hat{U}_i}$, and for all $w_i \in {W_i}, \hat{w}_i \in {\hat{W}_i}$}, 
	%\NN{As discussed before, I would use normal texts rather than symbols within a text environment. I made such improvements at some earlier points, but seems there are still too many of such things. Below you also see some more in red. Please check if there are still some other and rewrite them accordingly.} 
	we have for any $\hat{x}_i^+\in\hat{f}_i(\hat{x}_i,\hat{u}_i,\hat{w}_i)$: 
	\begin{align*}
		\mathcal{V}_i(x^+_i,\hat x^+_i)) \leq\mathcal{V}_i(x^+_i,f_i(\hat{x}_i, \hat{u}_i,\hat{w}_i)) +\hat{\gamma}_i(|\hat x^+_i - f_i(\hat{x}_i,\hat{u}_i,\hat{w}_i)|),
	\end{align*}
	where
	%\footnote{In this section, we assume that $\Sigma$ is deterministic.}
	$x^+_i = f_i(x_i,u_i,w_i)$.
	{By Definition~\ref{def:sym}, $\hat{x}^+_i \in \hat{f}_i(\hat{x}_i,\hat{u}_i,\hat{w}_i)$ implies $|\hat{x}^+_i - f_i(\hat{x}_i,\hat{u}_i,\hat{w}_i)|  \leq \eta^x_i$, thus, 	the above inequality reduces to}
	\begin{align*}
		\mathcal{V}_i&(x^+_i,\hat x^+_i)\leq\mathcal{V}_i(x^+_i,f_i(\hat{x}_i,\hat{u}_i,\hat{w}_i))+\hat{\gamma}_i(\eta^x_i).
	\end{align*}
	Observe that by \eqref{c2}, we obtain 
	\begin{align*}
		\mathcal{V}_i(x^+_i,f_i(\hat{x}_i,\hat{u}_i,\hat{w}_i))\leq\kappa_i(\mathcal{V}_i(x_i,\hat{x}_i))+\varrho_{w_i}(| w_i- \hat w_i|)+\varrho_{u_i}(| u_i- \hat u_i| ).
	\end{align*}	
	Hence, {for all $x_i\in {X}_i, \hat{x}_i \in {\hat{X}_i}$, for all $ u_i\in{{U}_i}, \hat{u}_i \in {\hat{U}_i}$, and for all $w_i \in {W_i}, \hat{w}_i \in {\hat{W}_i}$}, one obtains
	\begin{align} \label{veq}
		\mathcal{V}_i (x^+_i,\hat x^+_i) \leq \kappa_i(\mathcal{V}_i(x_i,\hat{x}_i))+\varrho_{w_i}(| w_i- \hat w_i|)+\varrho_{u_i}(| u_i- \hat u_i|)+\hat{\gamma}_i(\eta^x_i),
	\end{align}
	for any $\hat{x}_i^+\in\hat{f}_i(\hat{x}_i,\hat{u}_i,\hat{w}_i)$.
	Take any $x_i \in {X}_i$ and any $\hat x_i \in \hat {{X}}_i$ satisfying $\mathcal{V}_i(x_i,\hat x_i) \leq \varpi_i$, and any $w_i\in{{W}_i}$ and $\hat w_i \in \hat{{W}_i}$ such that $|w_i-\hat{w}_i|\leq \vartheta_i$. 
	For any $\hat{u}_i$, choose $u_i=\hat{u}_i$. Then, by combining \eqref{veq} with \eqref{secquantinit},
	we get that for ${x}_i^+ = {f}_i({x}_i,\hat{u}_i,{w}_i)$, there exists $\hat{x}_i^+ \in\hat{f}_i(\hat{x}_i,\hat{u}_i,\hat{w}_i)$ such that
	\begin{align}
		\mathcal{V}_i (x^+_i,\hat x^+_i) \leq \kappa_i(\varpi_i) + \varrho_{w_i}(\vartheta_i)+ \hat{\gamma}_i(\eta^x_i) \leq \varpi_i .
	\end{align}
	This implies that condition (ii) in Definition \ref{def:SFD1} is satisfied, and thus, $\V_i$ is a local $\varpi_i$-ASF from $\hat{\Sigma}_i$ to $\Sigma_i$. Similarly, we can also show that $\mathcal{V}_i$ is a local $\varpi_i$-ASF from $\Sigma_i$ to $\hat{\Sigma}_i$. In particular, by the structure of ${\hat{U}_i}=[ U_i]_{\eta^u_i}$, for any $u_i\in{{U}_i}$, there always exists $\hat{u}_i$ satisfying $|\hat{u}_i-u_i | \leq \eta^u_i$. As a result, for any  $\hat{x}_i^+ \in\hat{f}_i(\hat{x}_i,\hat{u}_i,\hat{w}_i)$, there exists 
	${x}_i^+ = {f}_i({x}_i,{u}_i,{w}_i)$ 
	%from~\eqref{secquantinit} and~\eqref{veq} 
	such that $\mathcal{V}_i (x^+_i,\hat x^+_i) \leq \kappa(\varpi_i) + \varrho_{w_i}(\vartheta_i)+ \varrho_{u_i}(\eta^u_i)+ \hat{\gamma}_i(\eta^x_i) \leq \varpi_i$.
	Therefore, we conclude that $\V_i$ is  a local $\varpi_i$-ASF both from $\hat{\Sigma}_i$ to $\Sigma_i$ and $\Sigma_i$ to $\hat{\Sigma}_i$.
\end{proof}

%Theorems~\ref{thm:3} and~\ref{thm:2} provide a compositional framework for the construction of symbolic models for an infinite network. However,
Given the results of Theorems~\ref{thm:3} and~\ref{thm:2}, one can observe that  inequalities~\eqref{compoquaninit} and~\eqref{secquantinit} are \emph{competing} conditions which may not hold simultaneously.
To resolve this issue, we propose a small-gain type condition ensuring the simultaneous satisfaction of both conditions.
%\NN{Be a bit more concrete! What does an ``inherent'' property mean? Are you saying that ``to resolve this issue, we propose a small-gain type condition ensuring the simultaneous satisfaction of both conditions? If so, then please use a ``direct'' language and state everything as clearly as possible.}
%and provide an algorithm design quantization parameters. 
\section{Compositionality Result}
In this section, we employ a small-gain type condition for the infinite network, under which one can always find suitable quantization parameters for the construction of symbolic models so that conditions~\eqref{compoquaninit} and~\eqref{secquantinit} are simultaneously satisfied.

Before stating the main result,  
%\NN{How many compositionality results do we have? From the headline of the section, it seems to me that there is only one. If it is the case, then why not simply saying ``the main result''...}, 
let us introduce the terminologies that will be used later.
In particular, we recall the notion of strongly connected components (SCCs) of graphs. 
We assume that the infinite network $\Sigma$ is composed of \emph{finitely} many sub-networks, where each of them is an \emph{infinite} network by itself, and the graph associated with each  sub-network is strongly connected \cite{baier2008principles}. 
%\NN{Please cite a well-known textbook on graph theory, where strong connectedness is discussed.}.

%	In this subsection, we exploit the interconnection topology of the network and employ knowledge from graph theory as an essential tool in the main results of this section. Here, we assume that the network is consisting of finitely many strongly connected components, each of which is composed of infinitely many subsystems. 
\subsection{Strongly connected components} 

Consider an infinite network $\Sigma=\mathcal{I}(\Sigma_i)_{i\in\N}$, as defined in Definition \ref{interconnectedsystem}. Hereafter, we denote the directed graph associated with $\Sigma=\mathcal{I}(\Sigma_i)_{i\in\N}$ by $G = (I, E)$, where $I = \N$ is the set of vertices with each vertex $i \in I$ labeled with subsystem $\Sigma_i$, and $E \subseteq I \times I$ is the set of ordered pairs $(i,j)$, $\forall i,j \in I$, with $y_{ji} \neq 0$. 		
Note that given the graph of our infinite network, we can formally define the finite index sets $\n_i$ and $\m_i$ of subsystem $\Sigma_i$, as mentioned in Subsection \ref{infdef}, i.e., $\n_i = \{j \in I| \exists (i,j) \in E\}$ and $\m_i = \{j \in I| \exists (j,i) \in E\}$. 

The SCCs of a directed graph $G$ are {maximal strongly connected subgraphs, i.e., no additional edges or vertices from G can be included in the subgraph without breaking its property of being strongly connected}~\cite{baier2008principles}. %\NN{What does it roughly speaking mean?}
{Given the structure of an infinite network, we denote by $\bar N \in \mathbb{N}$ the number of SCCs in the network.}
% we {simply denote the $\bar N$  maximal strongly connected sub-network by $SCC_k$}, $k \in [1;\bar N]$. 
%\NN{Confusing! What are trying to say here?}
In the sequel, we will denote the graphs of the SCCs in $G$ by $\bar G_k$, $k \in [1;\bar N]$, where $\bar G_k = (I_k, E_k)$ with $I_k = \N$. 
%Let $\bar\Sigma_k=\mathcal{I}(\Sigma_i)_{i\in I_k}$ denote the finite network related to the graph $G_k$. 
In addition, we define set $\n_{k_i} = \{j \in I_k| \exists (i,j) \in E_k\}$ 
which collects \emph{in-neighbors} of $\Sigma_i$ in $\bar G_k$, i.e., subsystems in the same subnetwork $\bar G_k$ who are \emph{directly} influencing $\Sigma_i$.
On the other hand, we define set $\m_{k_i} = \{j \in I_k| \exists (j,i) \in E_k\}$ which collects \emph{out-neighbors} of $\Sigma_i$ in $\bar G_k$, i.e., subsystems in the same subnetwork $\bar G_k$ that are \emph{directly} influenced by $\Sigma_i$.
Intuitively, $\n_{k_i}$ and $\m_{k_i}$ are the sets of neighboring subsystems of $\Sigma_i, i\in I_k$, in the same SCC.
%\NN{I don't quite agree with that. There should be at least one subsystem in $\n_{k_i}$ and/or $\m_{k_i}$ for some $k_i$, that is from another SCC.}
Note that if we regard each SCC in $G$ as a vertex, the resulting directed graph is \emph{acyclic}. 
We denote by $\tup{BSCC}(G)$ the collection of bottom SCCs $\bar G_k$ of graph $G$ from which no vertex in $G$ outside $\bar G_k$ is reachable. 
%		  The acyclic directed graph is denoted by $\bar G = (\bar I, \bar E)$, where $\bar I= \N$ is the set of vertices with each vertex $k \in \bar I$ labeled with $SCC_k$, equivalently $\Sigma_k$, and $\bar E = \{(k,l) \in \bar I\times \bar I \mid \exists i \in I_k, \exists j \in I_l, y_{ji} \neq 0\}$. Let $\bar\Sigma=\mathcal{I}(\bar\Sigma_k)_{k\in \bar I}$ denote the infinite network related to the graph $\bar G$. 
%		  Similarly, we denote by $\bar\n_k = \{l \in \bar I \mid \exists (k,l) \in \bar E\}$ the collection of SCCs in graph $G$ directly affect $G_k$, and by $\bar\m_k = \{l \in \bar I| \exists (l,k) \in \bar E\}$ the collection of SCCs in graph $\bar G$ that are directly affected by $G_k$. 
%		  In other words, $\bar\n_k$ and $\bar\m_k$, $k\in\bar I$ are finite subsets of $\bar I$ such that the index sets $\bar\n_k$ and $\bar\m_k$ enumerate the neighbors of $\bar\Sigma_k=\bar{\mathcal{I}}_k(\Sigma_i)_{i\in I_k}$, i.e., those finite networks $\bar\Sigma_l,l \in \bar\n_k$, $\bar\Sigma_{l'},l' \in \bar\m_k$ that affect or are affected by $\bar\Sigma_k$, respectively. By definition we require that $k\notin \bar\n_k\cup\bar\m_k$, $\forall k\in \bar I$. Moreover, each $\bar\Sigma_k$ can have only a finite number of neighbors by Definition \ref{subsys}.

In the next subsection, we leverage a small-gain type condition to facilitate the compositional construction of a symbolic model for an infinite network. 		  

\subsection{Small-gain theorem}
Consider an infinite network $\Sigma=\mathcal{I}(\Sigma_i)_{i\in\N}$ associated with a directed graph  $G$. Assume that each $\Sigma_i$ and its symbolic model $\hat{\Sigma}_i$ admit a local $\varpi_i$-ASF $\mathcal{V}_i$ with constants $\kappa_i, \rho_{w_i}, \underline \alpha_i \in \R_{>0}$ (as in Definitions \ref{def:SFD1} and \ref{iss}). Let $\bar G_k$, $k \in [1;\bar N]$, be the SCCs in $G$ with each $\bar G_k$ consisting of $\N$ vertices, where each vertex represents a subsystem. For any $\bar G_k$, we define for each $i,j \in \N$, 
\begin{align}\label{gammadcur}
	~\gamma_{ij} = \left \{ \begin{array}{cc} 
		(1-\kappa_{i})^{-1} \rho_{w_i}\underline \alpha_j^{-1}& \mbox{if } j \in \n_{k_i},\\
		0 & \mbox{otherwise}.
	\end{array}\right.
\end{align}
%where $\n_{k_i} \subset \N$ denotes the set of neighbors in $\bar G_k$ that directly affect subsysem $\Sigma_i$.
%	We assume that for every $\bar G_k$, $k \in [1;\bar N]$, the following small-gain type condition holds
%	\begin{align}\label{smallgaincur}
%	\gamma_{i_1i_2}\circ\gamma_{i_2i_3}\circ\cdots\circ\gamma_{i_{r-1}i_r}\circ\gamma_{i_ri_1}<\mathcal{I}_d,
%	\end{align}
%	$\forall(i_1,\ldots,i_r)\in\{k_1,\ldots,k_{\bar N_k}\}^r$, where $r\in \{1,\ldots,\bar N_k\}$.
%	
For each SCC $\bar G_k$, we introduce a gain operator $\Gamma_k:\ba_+ \rightarrow \ba_+$ by 
\begin{align}\label{Gam}
	\Gamma_k(s)=\big(\sup_{j\in\N}\{\gamma_{ij}s_j\}\big)_{i\in\N},\quad s\in\ba_+.
\end{align}

We furthermore assume that the following uniformity conditions hold for the constants introduced above.
\begin{assumption}\label{sgassump}
	There are constants $\overline \kappa$, $\overline{\rho}_{w}$, $\underline{\alpha} \in \R_{> 0}$, so that for all $i \in \N$ 
	\begin{align}
		\kappa_{i} \leq \overline \kappa,  \quad\quad  \rho_{w_i} \leq \overline{\rho}_{w}, \quad\quad \underline \alpha_i \geq \underline{\alpha}.
	\end{align}
\end{assumption}
Notice that the above assumption guarantees that the operator $\Gamma_k$ is well-defined. Accordingly, we have the following result recalled from \cite[Proposition 17]{mironchenko2020nonlinear}.
\begin{proposition}\label{sgprop}
	Under Assumption \ref{sgassump}, the following conditions are equivalent: 
	\begin{enumerate}[leftmargin=*]
		\item[(i)] The spectral radius of $\Gamma_k$ satisfies
		\begin{align}\label{SGC1}
			r(\Gamma_k) = \lim\limits_{n\rightarrow \infty} \big( \sup\limits_{j_1,\dots,j_{n+1} \in \N}  \gamma_{j_1j_2}\dots \gamma_{j_nj_{n+1}}  \big)^{1/n} <1.
		\end{align}
		\item[(ii)] There exist a vector $\sigma_k  \in \textup{int}(\ba_+)$ and constant $\lambda_k \in (0,1)$ satisfying
		\begin{align}\label{SGC2}
			\Gamma_k(\sigma_k) \leq \lambda_k \sigma_k. 
		\end{align} 
	\end{enumerate}
\end{proposition}

The following theorem states the main result of this section. 
\begin{theorem} \label{Main}
	Consider a network $\Sigma=\mathcal{I}(\Sigma_i)_{i\in\N}$. 
	Suppose that Assumption \ref{sgassump} holds. Assume that for each SCC in $\Sigma$, condition~\eqref{SGC1} holds. Then, 
	for any desired precision $\varpi \in \mathbb{R}_{>0}$ as in Definition \ref{SF}, and for each $i \in \N$,
	there exist quantization parameters  $\eta^x_i$, $\eta^u_i$, $\phi_i$, as designed in Algorithm \ref{quantialgo}, such that~\eqref{compoquaninit} and~\eqref{secquantinit} are satisfied simultaneously.
	%	where the local parameters $\vartheta_i \in  \mathbb{R}_{>0}$ and $\varpi_i  \in  \mathbb{R}_{>0}$,  $\forall i  \in \N$, are obtained from Algorithm \ref{quantialgo}. 	
	\begin{proof}
		Note that by Proposition \ref{sgprop}, the spectral radius condition \eqref{SGC1} implies that for each $\bar G_k$, there exists a vector $\sigma_k = (\sigma_{k_i})_{i\in \N}$ satisfying \eqref{SGC2}. Hence, we get 
		\begin{align} \label{imeq}
			\Gamma_k(\sigma_k) = \big(\sup_{j\in\N}\{\gamma_{ij}\sigma_{k_{j}}\}\big)_{i\in\N} \leq \lambda_k\sigma_k 
			\Longrightarrow \sup_{j\in\N}\{\gamma_{ij}\sigma_{k_{j}}\} \leq \lambda_k \sigma_{k_{i}} < \sigma_{k_{i}}. 
		\end{align}
		Since \eqref{imeq} holds for all $i \in \N$, one has
		\begin{align} \notag
			&\sup\limits_{j\in\N}\{\gamma_{ij}\sigma_{k_{j}}\}<\sigma_{k_{i}}
			\stackrel{\eqref{gammadcur}}\Longrightarrow \sup\limits_{j\in\N}\{(1-\kappa_{i})^{-1} \rho_{w_i}\underline \alpha_j^{-1}\sigma_{k_{j}}\} < \sigma_{k_{i}} \\  \label{inequalityinscc}
			&\Longrightarrow  \rho_{w_i} \max\limits_{j \in \n_{k_i}}\{\underline \alpha_j^{-1} \sigma_{k_{j}}\} <  (1-\kappa_{i})\sigma_{k_{i}}.
		\end{align}
		Now, set $\varpi_{k_i}=\sigma_{k_i}r$, for all $i \in \N$, where $r\in \R_{> 0}$ is chosen under the criteria given in  lines 5 and 7 of Algorithm~\ref{quantialgo}. Choose the internal input quantization parameters $\phi_{ij}$ such that for all $i\in \N$ \vspace{-0.2cm} 
		\begin{align} \label{wquanti}
			\max\limits_{j \in \n_{k_i}}\{\phi_{ij}\} < \rho_{w_i}^{-1} (1-\kappa_{i})\varpi_{k_i}-\max\limits_{j \in \n_{k_i}}\{\underline \alpha_j^{-1}\varpi_{k_j}\}.
		\end{align}
		By setting $\vartheta_i \!=\! \max\limits_{j \in \n_{k_i}}\{\underline \alpha_j^{-1}\varpi_{k_j} \!+\!\phi_{ij}\}$ and combining with  \eqref{wquanti}, one has, for all $i \!\in\! \N$,
		\begin{align} \notag
			%\kappa_{i}(\varpi_i)-\rho_{inti}(\vartheta_i) > 0,
			&\rho_{w_i}\vartheta_i = \rho_{w_i}\max\limits_{j \in \n_{k_i}}\{\underline \alpha_j^{-1}\varpi_{k_j} +\phi_{ij}\}\\ \notag
			&\leq \rho_{w_i} (\max\limits_{j \in \n_{k_i}}\{\underline \alpha_j^{-1}\varpi_{k_j} \}+\max\limits_{j \in \n_{k_i}}\{\phi_{ij}\}) \stackrel{\eqref{wquanti}}< (1-\kappa_{i})\varpi_{k_i},
		\end{align}
		which implies that one can always find suitable local quantization parameters  $\eta^x_i$ and  $\eta^u_i$  to satisfy \eqref{secquantinit}.
		%	It can be observed that, in the procedure done in lines 8-10 and 16-19, the chosen value of the temporary variable $r \in \mathbb{R}_{\geq0}$ depends on the desired precision $\varpi$ and the hierarchy of the directed graph.
		Additionally,  
		%The parameters $(\varpi_i, \vartheta_i)$ are designed in each SCC. 	Note that in lines 8-9 (resp. lines 15-16), the value of the temporary variable $r \in \mathbb{R}_{\geq0}$ is chosen within each SCC, depending on the desired precision $\varpi$ (resp. the successors of the current SCC). 
		%	Consequently, by employing this algorithm, together with \eqref{wquanti} and \eqref{inequalityinscc}, one has
		%	\begin{align}\label{posiinequal}
		%	%\kappa_{i}(\varpi_i)-\rho_{inti}(\vartheta_i) > 0,
		%{\rho_{inti}(\vartheta_i) = \rho_{inti}(\max\limits_{j \in \tup{Pre}_{I_k}(i)} \{\alpha^{-1}_j(\varpi_j)  + \phi_{ij}\})\leq \rho_{inti} \circ\max\limits_{j \in \tup{Pre}_{I_k}(i)}\{\alpha^{-1}_j(\varpi_j)+ \bar \alpha_{j}^{-1}(\varpi_j)\}< \kappa_{i}(\varpi_i),}
		%	\end{align}
		%	As a result, one can ensure that 
		%	\begin{align}\label{posiinequal}
		%	\kappa_{i}(\varpi_i)-\rho_{inti}(\vartheta_i) > 0,
		%	\end{align}
		%	Thus, given any pair of parameters $(\varpi_i, \vartheta_i)$ as obtained by Algorithm \ref{quantialgo}, by \eqref{posiinequal}, one can always find suitable local parameters $\eta_i, \mu_i$ to satisfy \eqref{secquantinit} (resp. \eqref{secquant}). 
		the selection of $\vartheta_i = \max\limits_{j \in \n_{k_i}}\{\underline \alpha_j^{-1}\varpi_{k_j} +\phi_{ij}\}$ as in line 9 of Algorithm \ref{quantialgo}, together with the design procedure for $\varpi_i$ and $\phi_{ij}$ ensure that  \eqref{compoquaninit}  is satisfied as well, which concludes the proof.
	\end{proof}
\end{theorem}
\IncMargin{0.02em}
\begin{algorithm2e*}[t!]
	\DontPrintSemicolon
	\SetAlgoNoLine
	%\Indm 
	\KwInput{The desired precision $\varpi \in \mathbb{R}_{>0}$; the directed graph $G$ composed of SCCs $\bar G_k$, $\forall k \in [1;\bar N]$, and vectors $\sigma_k = (\sigma_{k_i})_{i\in \N}$ satisfying \eqref{SGC2} for $\bar G_k$; the functions $\mathcal{V}_i$ equipped with $\kappa_i, \rho_{w_i}, \underline \alpha_i \in \R_{\geq 0}$, $\forall i \in\N$.}	
	\KwOutput{$\eta^x_i, \eta^u_i, \phi_i  \in \R_{\geq 0}$, $\forall i \in\N$.}
	%\Indp	
	Set $\varpi_i := \infty$, $\vartheta_i := \infty$, $\forall i \in\N$, $\forall k \in [1;\bar N]$, $G^* =  G$; 
	
	\While{$G^* \neq \varnothing$}
	{
		\ForEach{$\bar G_k \in \tup{BSCC}(G^*)$}
		{	
			\eIf{$G^* =  G$}
			{ 
				\Indm 				
				choose $r \in \mathbb{R}_{>0}$ s.t. $\sup\limits_{i \in\N}\{\sigma_{k_i}r\}= \varpi$;
				
			}
			{
				\Indm 
				choose $r \in \mathbb{R}_{>0}$ s.t. $\sigma_{k_i}r \leq \underline\alpha_{i}\min\limits_{j \in \m_{i}\backslash\m_{k_i}}\{\vartheta_{j}-\phi_{ji}\}$, $\forall i \in\N$; 
				%with $\m_{i} \backslash \m_{k_i} \neq \varnothing$;	
			}				
			set $\varpi_{k_i}\!=\!\sigma_{k_i}r$, choose $\phi_{ij}$, $\forall i,j \in\N$, s.t. $\max\limits_{j \in \n_{k_i}}\{\phi_{ij}\} \!<\! \rho_{w_i}^{-1} \kappa_{i}\varpi_{k_i}\!-\!$ $\max\limits_{j \in \n_{k_i}}\{\underline \alpha_j^{-1}\varpi_{k_j}\}$; set $\vartheta_{k_i} =$ $\max\limits_{j \in \n_{k_i}}\{\underline \alpha_j^{-1}\varpi_{k_j}+\phi_{ij}\}$, $\forall i \in\N$;				
			choose $\phi_{ij} <\vartheta_i$, $\forall i \in\N$, $\forall j\in \n_{i}\backslash\n_{k_i}$;	
		}
		
		$G^* = G^* \backslash \tup{BSCC}(G^*)$;
	}
	Compute $\eta^x_i$ and  $\eta^u_i$ s.t.	$\eta^x_i \leq   \hat{\gamma}_i^{-1}[(1-\kappa_i) \varpi_i - \rho_{w_i}\vartheta_i-\rho_{u_i}(\eta^u_i)]$, $\forall i \in\N$.
	\caption{Compositional design of local quantization parameters $\eta^x_i, \eta^u_i, \phi_i \in \mathbb{R}_{\geq 0}$, $\forall i \in\N$} \label{quantialgo}
\end{algorithm2e*}
\DecMargin{0.02em}
\begin{remark}\label{nosgc}
	Note that if $\gamma_{ij}< 1$ for any $i,j\in\N$, the  spectral radius condition $r(\Gamma_k) <1$ as in \eqref{SGC1} is satisfied automatically.  In this case, by Proposition \ref{sgprop}, there always exists  $\lambda_k \in (0,1)$ such that inequality \eqref{SGC2} holds with $\sigma_{k} = (1)_{i \in \N}$ and $\sup_{i\in\N}\{\gamma_{ij}\} \leq \lambda_k$.
	Note that by involving the notion of SCCs in the design procedure for the selection of parameters, we are allowed to check the small-gain condition and design local quantization parameters inside each SCC, independently of the entire network. 
	%Let us also remark the soundness of Algorithm \ref{quantialgo}. It can be seen that as long as Assumption \ref{smallgainSCC} holds, for any desired precision $\varpi \in \mathbb{R}_{>0}$, the algorithm always provides us suitable pairs of parameters $(\varpi_i, \vartheta_i)$, $\forall i\in [1;N]$. 
	In addition, since the original infinite network is composed of a finite number of SCCs, the algorithm terminates in \emph{finite} iterations. \hfill$\diamond$ 
\end{remark}

\subsection{Safety controllers}

In this subsection, we consider a safety synthesis problem for {an} infinite network. Note that classical safety synthesis methods are not applicable any more in this context since they require infinite memory.
Here, we show a compositional approach which addresses {such a} synthesis problem in a decentralized manner.  
%Suppose the state constraint for each subsystem $\Sigma_i$ is given by the compact safe set $X_i \subseteq \R^{n_i}$ and admissible input set is denoted by $U_i \subseteq \R^{m_i} $. 
%Then for the interconnected system $\Sigma$, the safe set $X$ and input constraint set $U$ have the following structure 
%\begin{align} 	\label{eq:safeset}
%	X = \prod\limits_{i=1}^N X_i, {\kern 1pt}{\kern 1pt} {\mathop{\rm with}} {\kern 1pt}{\kern 1pt} X_i \subseteq \R^{n_i}, \sum_{i=1}^{N}n_i = n,\\\vspace{-1.5mm}
%	\label{eq:inputset}
%	U = \prod\limits_{i=1}^N U_i, {\kern 1pt}{\kern 1pt}{\mathop{\rm with}} {\kern 1pt}{\kern 1pt} U_i \subseteq \R^{m_i}, \sum_{i=1}^{N}m_i = m.
%\end{align}

%
%\begin{align*}
%X &= \{ x = (x_i)_{i\in\N} : x_i \in X_i,\ \Vert x\Vert:=\sup_{i\in\N}\{|x_i|\} < \infty \},\\
%U &= \{ u = (u_i)_{i\in\N} : u_i \in U_i,\ \Vert u\Vert:=\sup_{i\in\N}\{|u_i|\}< \infty \},\\
%f(x,u)&=\{(x^+_i)_{i\in\N}|x_i^+\in f_i(x_i,w_i,u_i)\},\\ 
%Y&=\prod_{i\in\N}Y_{ii},\\	
%h(x) &= (h_{ii}(x_i))_{i\in\N}.
%\end{align*} 
Consider an infinite network $\Sigma = (X,U,f,Y,h)$ as in Definition \ref{interconnectedsystem}, consisting of subsystems $\Sigma_i = (X_i,U_i,W_i,f_i,$ $Y_i,h_i)$, $i\in\N$, as in \eqref{eq:2}. 
%Now consider an infinite network $\Sigma $ as in Definition \ref{interconnectedsystem} with subsystems $\Sigma_i$ as in \eqref{eq:2} and 
Suppose we are given a global decomposable safety specification  $S=\prod_{i\in\N}S_{i}$.
% {with $S_i$ as in Definition~\ref{safetycontroller}}. 
We define $Out = {\ba} \setminus S$ and its projection on the $i$-th subsystem as $Out_i = {\mathbb{R}^{n_i}} \setminus S_i$.
From Definition~\ref{interconnectedsystem}, the state transition function of the infinite network holds the following relations:\\
For all $x = (x_i)_{i\in\N} \in S$, all $u = (u_i)_{i\in\N} \in U$, all $x' = (x_i')_{i\in\N} \in S$, 
\begin{align} \label{evo1}
	x' \in f(x,u) \Longleftrightarrow x_{i}' \in f_i(x_i,u_i,w_i), w_{ij} = h_{ji}(x_j), \forall i \in \N, \forall j\in\n_i.    
\end{align}	
%where $x'$ is the successor state from state $x$ under external input $u$.\\
For all $x = (x_i)_{i\in\N} \in S$, all $u = (u_i)_{i\in\N} \in U$, 
\begin{align} \label{evo2}
	Out \cap \{f(x,u)\} \neq \varnothing \Longleftrightarrow 
	\exists i \in \N: Out_i \cap \{f_i(x_i,u_i,w_i)\} \neq \varnothing, w_{ij} = h_{ji}(x_j),\forall j \in \n_i.   
\end{align}

Now, we introduce the notion of safety controllers that are used to enforce safety specifications over the subsystems.
%Now we provide the formal definition of safety controllers.

\begin{definition}
	\label{safetycontroller}
	A safety controller for a discrete-time control subsystem $\Sigma_i$ and the safe set $S_i \subseteq X_i$ is a map $C_i: X_i \rightrightarrows U_i$ such that:	\begin{enumerate}
		%	\item $\forall x \in \R^n$, $C(x) \subseteq U $;
		\item[{(i)}] $\tup{dom}(C_i) = \{ x_i\in X_i |C_i(x_i) \neq \varnothing\} \subseteq S_i$;
		\item[{(ii)}] $f_i(x_i,u_i,w_i) \subseteq \tup{dom}(C_i)$ for all $x_i\in \tup{dom}(C_i)$, all $u_i\in C_i(x_i)$, and all $w_i\in W_i$.
	\end{enumerate}	 
\end{definition}
\begin{remark}
	Note that the safety controllers for subsystems are synthesized by following an assume-guarantee reasoning \cite{henzinger1998you}. In particular, for each subsystem $\Sigma_i$, we guarantee that safety controller $C_i$ (if existing) enforces the safety specification  $S_i$ over $\Sigma_i$, by assuming that all of its in-neighbors $\Sigma_j$, $j \in \n_i$, have safety controllers $C_j$ enforcing safety  specifications $S_j$.
	%evolve within their own safe sets $S_j$. 
	Moreover, since the interconnection variables of the concrete infinite network are constrained as $w_{ij} = y_{ji}$ (cf. Definition \ref{interconnectedsystem}), for all $i\in\N, j\in\n_i$,  the internal input set $W_i$ considered in Definition \ref{safetycontroller} is restricted to $W_i=\prod_{j\in{\n_i}}W_{ij} = \prod_{j\in{\n_i}}Y_{ji}$ where 
	$Y_{ji} = \{h_{ji}(x_j)| x_j \in S_j\}$. 
	%$y_{ji} =h_{ji}(x_j), x_j \in S_j$.}		
	\hfill$\diamond$ 
\end{remark}

Similarly, the definition of a safety controller for the \emph{overall} network is given as follows.
\begin{definition}
	\label{safetycontrollerinf}
	A safety controller for an infinite network $\Sigma$ and the safe set $S \subseteq X$ is a map $C: X \rightrightarrows U$ such that:
	\begin{enumerate}
		%	\item $\forall x \in \R^n$, $C(x) \subseteq U $;
		\item[{(i)}] $\tup{dom}(C) = \{ x\in X |C(x) \neq \varnothing\} \subseteq S$;
		\item[{(ii)}] {$f(x,u) \subseteq \tup{dom}(C)$ for all $x\in \tup{dom}(C)$ and all $u\in C(x)$}.
	\end{enumerate}	 
\end{definition}
%
%\begin{remark}
%	Note that the definition of a safety controller for an infinite network $\Sigma_i$ is similar to that of Definition \ref{safetycontroller}, and the slight modification lies in condition 2 where the state transition function has to be modified to $f_i(x_i,u_i,z_i,w_i)$.
%\end{remark}

%It is known \cite{tabuada2009verification} that, there exists a maximal safety controller $C^*$ for control system $\Sigma$ and safe set $X$ containing all safety controllers, i.e., $C(x) \subseteq C^*(x)$ for all $x \in \R^n$. This maximal safety controller can be computed theoretically using the well-known fixed-point algorithm \cite{kerrigan2001robust}.

Suppose {that} we are given local safety controllers $C_i$ as in Definition \ref{safetycontroller} for all $i \in \N$, each corresponding to subsystems $\Sigma_i$ and safety specification $S_i$.  
Let controller $C: X \rightrightarrows U$ be defined by $C(Out) = \varnothing$ and
\begin{align}\label{CompositionalController} 	
	\forall i \in \N {\kern 1pt}{\kern 1pt}{\kern 1pt}{\kern 1pt} {\mathop{\rm with}} {\kern 1pt}{\kern 1pt}{\kern 1pt}{\kern 1pt} x_i \in X_i, C(x) = \{u  \in U |u_i \in C_i(x_i), {i \in \N} \},    
\end{align}
where $x = (x_i)_{i\in\N} \in X$, $u = (u_i)_{i\in\N} \in U$.

Now, we provide the next proposition, adapted from \cite[Theorem 3.1]{liu2019compositional}, which shows that the composed controller as defined above works for the overall infinite network. 
%The next theorem is adapted from \cite{liu2019compositional}.
\begin{proposition}
	\label{Theoremcontroler}
	Controller $C: X \rightrightarrows U$ defined in \eqref{CompositionalController} is a safety controller for the infinite network $\Sigma$ and safe set $S$.
\end{proposition}
\begin{proof}
	We start by showing condition (i) of Definition \ref{safetycontrollerinf}.
	By~\eqref{CompositionalController}, it can be readily seen that {for all} $x \in X$ with $C(x) \neq \varnothing$ {we get} $x \notin Out$.
	From the  definition of $Out = {\ba} \setminus S$, we {have that all $x \in X$ where $C(x) \neq \varnothing$ necessarily lie inside $S$}, which satisfies condition (i) of Definition \ref{safetycontrollerinf}.	
	We proceed to show condition (ii) of Definition \ref{safetycontrollerinf}. 
	Let $x \in \tup{dom}(C) \subseteq S$, $u \in C(x)$ and $x' \in f(x,u)$.
	First, we show $x' \in S$ by contradiction. If $x' \notin S$, then $x' \in Out$.
	From \eqref{evo2}, there exists  $i \in \N$, such that $Out_i \cap \{f_i(x_i,u_i,w_i)\} \neq \varnothing$, which contradicts the fact that $u_i \in C_i(x_i)$ with $C_i$ being the safety controller for subsystem $\Sigma_i$ and the corresponding safe set $S_i$. 
	Therefore, we have $x' \in S$. 
	From \eqref{evo1}, it is clear that, for each $i \in \N$, $x'_i \in f_i(x_i,u_i,w_i)$. 
	Moreover, {by condition (ii) of Definition \ref{safetycontroller}, $u_i \in C_i(x_i)$ implies that $x'_i \in dom(C_i)$.} 
	For all $i \in \N$, let $u'_i \in C_i(x'_i)$ and by \eqref{CompositionalController}, we have $u' = (u_i')_{i\in\N} \in C(x')$ and $x' \in dom(C)$. {It follows that condition (ii) of Definition \ref{safetycontrollerinf} is satisfied as well.} 
	Hence, we conclude that $C$ is a safety controller for $\Sigma$ and safe set $S$.
\end{proof}

Proposition~\ref{Theoremcontroler} shows that one can obtain a global safety controller for an infinite network which enforces an overall safety specification by composing local safety controllers designed for subsystems. 
In that way, 
%once a symbolic model of the concrete network is achieved, 
one can follow this decentralized controller synthesis strategy to
easily design local safety controllers for the local symbolic models, and then refine the controllers back to the concrete subsystems via the corresponding alternating simulation relations across them.

\section{Case Study}\label{casestudy}
In this section, we present our results on 
%demonstrate the effectiveness of the proposed results to 
a road traffic network divided into infinitely many road cells. We first construct a symbolic model of the infinite network in a compositional way. 
Then we use the constructed symbolic model as a substitute to compositionally synthesize a safety controller to keep the density of traffic in each cell remaining within a desired region. 
The effectiveness of our results is also shown in comparison with the existing compositional results in~\cite{SwikirIfac}.
\subsection{Road traffic network}
In this subsection, let us first introduce the model of this case study which is a variant of the road traffic model in \cite{canudas2012graph}.  
Here, the traffic flow model is considered as a network divided into infinitely many cells.  Each cell, indexed by $i \in \N$, can be modeled as a one-dimensional subsystem, represented as a tuple $\Sigma{_i} = (X{_i},U{_i},W{_i},f{_i},X{_i},\id)$. Moreover, each cell is assumed to be equipped with at least one measurable entry and one exit. The traffic flow dynamics of each cell is given by
\begin{align}\label{tsubsys}
	\Sigma_i:\left\{
	\begin{array}{rl}
		\mathbf{x}_i(k+1)= & (1-\frac{\tau v}{l}-e)\mathbf{x}_i(k) + d_i\omega_i(k)+b\nu_i(k),\\
		\mathbf{y}{_i}(k)=&\mathbf{x}{_i}(k),
	\end{array}
	\right.
\end{align}
where $\tau$ is the sampling time in hour, $l$ is the length of each cell in kilometers, and
$v$ is the traffic flow speed in kilometers per hour. For each cell $i \in \N$ in the network, the state $\mathbf{x}_i(k)$ of each subsystem $\Sigma_i$ represents the density of the traffic in vehicle per cell at a specific time instant indexed by $k$. The scalar $b$ denotes the number of vehicles that are allowed to enter each cell during each sampling time controlled by the input signals $\nu_i(\cdot) \in \{0,1\}$, where $\nu_i(\cdot) = 1$ (resp. $\nu_i(\cdot) = 0$) corresponds to green (resp. red) traffic light. The constant $e$ denotes the percentage of vehicles that leave the cell during each sampling time through exits.  
\begin{figure}
	%\begin{center}
		\includegraphics[width=0.8\textwidth]{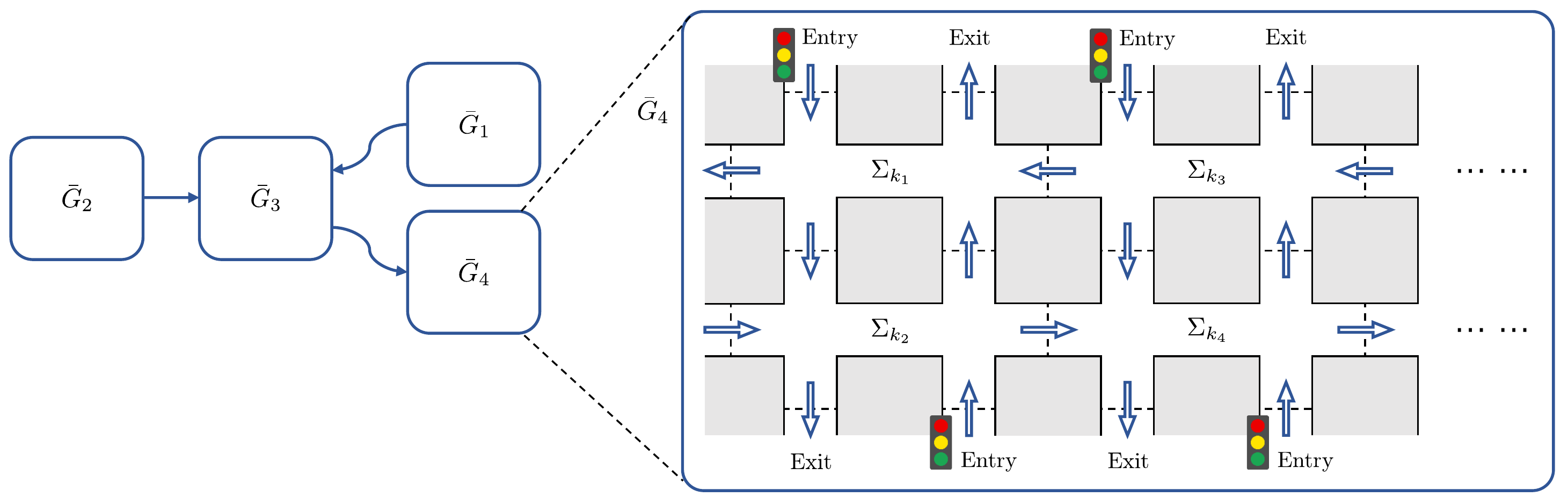}
		\vspace{-0.3em}
		\caption{Model of a road traffic network composed of four subnetworks, each of which consists of infinitely many subsystems.}
		\label{infnet}
	%\end{center}
	\vspace{-0.3cm}
\end{figure}
The left side of Figure~\ref{infnet} shows the structure of the traffic network as a directed graph  consisting of $\bar N = 4$ strongly connected subnetworks, each of which is denoted by $\bar{G}_k$, $k\in\{1,2,3,4\}$. Subnetworks are connected through single-directional freeways.
The right side of Figure~\ref{infnet} roughly depicts the traffic network topology of subnetwork $\bar G_1$ consisting of infinitely many cells (modeled by $\Sigma_i$) with different link models. %which are the intersection of links (directional segments of freeways).
The internal inputs of the subsystems satisfy the following interconnection structure:
%\begin{align}
\begin{enumerate}
	\item[(i)] For subsystems $\Sigma_{i}$ in subnetworks $\bar{G}_1$ and  $\bar{G}_2$
	\begin{enumerate}
		\item[$\bullet$] $d_i=(\frac{1-e}{2})(\frac{\tau v}{l},\frac{\tau v}{l})\trn$,  $\omega_i=(\mathbf{y}_{i+1},\mathbf{y}_{i+2})$ if $i\in \{2c+1: c \in \N_0\}$;
		%\item[$\bullet$] $d_i=(1-e)\frac{\tau v}{l},\omega_i=\mathbf{y}_{i+1}$ if $i\in S_2:=\{2\}$;
		\item[$\bullet$] $d_i=(1-e)\frac{\tau v}{l},\omega_i=\mathbf{y}_{i-1}$ if $i\in \{2\}$;
		\item[$\bullet$] $d_i=(\frac{1-e}{2})(\frac{\tau v}{l},\frac{\tau v}{l})\trn$,  $\omega_i=(\mathbf{y}_{i-2},\mathbf{y}_{i-1})$ if $i\in \{2c+2: c \in \N\}$.
		%\{4+2c\}$;
	\end{enumerate}	
	\item[(ii)] For subsystems $\Sigma_{i}$ in subnetworks $\bar{G}_3$ and  $\bar{G}_4$
	\begin{enumerate}
		\item[$\bullet$] $d_i=(\frac{1-e}{2})(\frac{\tau v}{l},\frac{\tau v}{l})\trn$, $\omega_i=(\mathbf{y}_{i+1},\mathbf{y}_{i+2})$ if $i\in \{2c+1: c \in \N_0\}$;
		%\item[$\bullet$] $d_i=(1-e)\frac{\tau v}{l},\omega_i=\mathbf{y}_{i+1}$ if $i\in S_2:=\{2\}$;
		\item[$\bullet$] $d_i=(\frac{1-e}{2})(\frac{\tau v}{l},\frac{\tau v}{l})\trn$, $\omega_i=(\mathbf{y}_{i-2},\mathbf{y}_{i-1})$ if $i\in \{2c+2: c \in \N_0\}$,
		%\{4+2c\}$;
	\end{enumerate}	
\end{enumerate}
where $\mathbf{y}_0 = \mathbf{y}_n$, $n \in I_{k-1}, k \in \{3,4\}$.
By Definition \ref{interconnectedsystem}, the infinite network $\Sigma=\mathcal{I}(\Sigma_i)_{i\in\N}$ is denoted by a tuple $\Sigma = (X,U,f,X,\id)$, where 
$X = \{ x = (x_i)_{i\in\N} : x_i \in X_i\}$, $U = \{ u = (u_i)_{i\in\N} : u_i \in U_i\}$, $f(x,u) =\{(x^+_i)_{i\in\N}|x_i^+\in f_i(x_i,u_i,w_i)\}$, and $Y =\prod_{i\in\N}X_{i}$.
First we show the well-posedness of the overall network by establishing that 
%Let us first show that the infinite network is well-posed by showing that
$\Vert f(x,u) \Vert < \infty$. Note that we have 
\begin{align*}
	\Vert f(x,u)\Vert&=\sup\limits_{i\in \N}\{|f_i(x_i,u_i,w_i)|\}\stackrel{\eqref{tsubsys}}=\sup\limits_{i\in \N}\{|(1-\frac{\tau v}{l}-e)x_i + d_iw_i + bu_i|\}\\
	&\leq |(1-\frac{\tau v}{l}-e)|\sup\limits_{i\in \N}\{|x_i| \}+|(1-e)\frac{\tau v}{l}|\sup\limits_{i\in \N}\{|x_i |\}+|b|\sup\limits_{i\in \N}\{|u_i |\}\\
	&\leq \max\{|(1\!-\!\frac{\tau v}{l}\!-\!e)|,|(1\!-\!e)\frac{\tau v}{l}|,|b|\}(\sup\limits_{i\in \N}\{|x_i| \}\!+\!\sup\limits_{i\in \N}\{|x_i|\}\!+\!\sup\limits_{i\in \N}\{|u_i |\})\\
	&\stackrel{\eqref{infnorm}}=\max\{|(1\!-\!\frac{\tau v}{l}\!-\!e)|,|(1\!-\!e)\frac{\tau v}{l}|,|b|\}(2\Vert x\Vert+\Vert u\Vert\})<\infty.
\end{align*}
Therefore, the infinite network $\Sigma=\mathcal{I}(\Sigma_i)_{i\in\N}$ is well-posed. Moreover, each subsystem admits a $\delta$-ISS Lyapunov function of the form $\V_i(x_i,\hat{x}_i) = |x_i-\hat{x}_i| $ satisfying conditions \eqref{c1}--\eqref{tinq} for all $i \in \N$ with ${\underline{\psi}_i}= {\overline{\psi}_i} = \id$, $\kappa_i = (1-\frac{\tau v}{l}-e)\id$, $\rho_{w_i} = |(1-e)\frac{\tau v}{l}|\id$, and $\rho_{u_i}= \hat \gamma_i = \id$.

\subsection{Hierarchical compositional construction of symbolic model}
Now set the parameter values of the system as $\tau=\frac{10}{60\times 60}$h, $l=0.5$km, $v=60$km/h, $b=5$, and $e = 0.1$. 
We construct a symbolic model that simulates the infinite network through an $\varpi$-ASF as in Definition $\ref{SF}$.  
For a given desired parameter $\varpi$, the output behavior of the constructed symbolic network will mimic that of the original network with a mismatch $\hat \varepsilon = \alpha^{-1}(\varpi)$ (cf. Remark \ref{remarkaccuracy}). 
By fixing $\varpi = 0.8$, we apply our compositionality results to design proper quantization parameters for all the subsystems, so that the overall symbolic network simulates the original infinite network with precision $\hat \varepsilon$.
First note that for each strongly connected subnetwork $\bar{G}_k$, 
by \eqref{gammadcur}, it can be verified that $\gamma_{ij} <1$, for all $i \in \N$, $j \in \n_{k_i}$, and the uniformity conditions in Assumption \ref{sgassump} hold readily. Thus, the spectral radius condition \eqref{SGC1} is satisfied, and condition \eqref{SGC2} holds  as well with a candidate vector $\sigma_{k} = (\sigma_{k})_{i \in \N} = (1)_{i \in \N}$ (cf. Remark \ref{nosgc}). 
Next, given the desired parameter $\varpi$,
we apply Algorithm \ref{quantialgo} to design local quantization parameters compositionally. 
We start with $G^* = G$ and get the bottom strongly connected subnetwork $\tup{BSCC}(G^*) = {\bar G}_4$ for line 3 in Algorithm \ref{quantialgo}. 
Consider the subnetwork ${\bar G}_4$, we choose $r = \varpi = 0.8$, $\phi_{ij} = 0$, and accordingly $\varpi_{k_i} = \vartheta_{k_i} = r$ so that the conditions in lines 5 and 9 are satisfied.  Now $G^*$ is updated in line 11 to $\{{\bar G}_1,{\bar G}_2,{\bar G}_3\}$ and the BSCC of the updated $G^*$ is ${\bar G}_3$. 
We proceed by choosing $r = \min_{j \in I_4}\vartheta_{j} = 0.8$ to satisfy the conditions in lines 7 and 9 with $\varpi_{k_i} = \vartheta_{k_i} = 0.8$ and $\phi_{ij} = 0$. 
Now $G^*$ and its BSCCs are updated to $G^* = \tup{BSCC}(G^*) = \{{\bar G}_1, {\bar G}_2\}$. Similarly, one can choose $\varpi_{k_i} = \vartheta_{k_i} = 0.8$ and $\phi_{ij} = 0$ for all of the subsystems in subnetworks ${\bar G}_1$ and ${\bar G}_2$ such that conditions in lines 7 and 9 are satisfied. 
Till here, we obtain local parameters $(\varpi_i, \vartheta_i) = (0.8,0.8)$ for all $i \in \N$. 
Now we proceed to design local quantization parameters $\eta^x_i$ and $\eta^u_i$ such that the inequality in line 13 holds with the parameters $(\varpi_i, \vartheta_i)$ we just obtained. 
Here, we take the local quantization parameters as $\eta^x_i = 0.1$ and $\eta^u_i = 0$, for all $i \in \N$, which will be later used to build local symbolic models of all the subsystems. Using the result in Theorem \ref{thm:2}, one can readily verify that the $\delta$-ISS Lyapunov function $\V_i(x_i,\hat{x}_i) = |x_i-\hat{x}_i|$ is a local $\varpi_{i}$-ASF from each local symbolic model $\hat \Sigma_i$ to the original subsystem $\Sigma_i$. Furthermore, by Theorem \ref{thm:3}, $\tilde{\V}(x,\hat{x}) = \sup_{i \in \N}\{|x_i-\hat{x}_i|\}$ is well-defined and is an $\varpi$-ASF from the abstract network $\hat \Sigma=\mathcal{I}(\hat\Sigma_i)_{i\in\N}$ to the original infinite network $\Sigma=\mathcal{I}(\Sigma_i)_{i\in\N}$.
We have the guarantee that the mismatch between the output behaviors of the infinite network $\Sigma$ and that of its symbolic model $\hat \Sigma$ will not exceed $\hat \varepsilon = \alpha^{-1}(\varpi) = 0.8$ (cf. Remark \ref{remarkaccuracy}).

Here, let us compare our compositional technique with the one proposed in \cite{SwikirIfac}.  
%Note that the infinite road traffic network model was also adopted in \cite{SwikirIfac} to illustrate the compositional abstraction technique proposed therein. 
Note that the same traffic network model was also adopted in \cite{SwikirIfac} to illustrate the compositional abstraction technique proposed there. 
Using same state and input quantization parameters $\eta^x_i = 0.1, \eta^u_i = 0$ as in the present paper, 
%the compositionality framework proposed in \cite{SwikirIfac}, 
the overall approximation error between related networks obtained in \cite{SwikirIfac} is $\hat \varepsilon = 1.7$, which is much larger than the one we obtained here ($\hat \varepsilon = 0.8$ as computed in the last paragraph). The reason is due to the conservatism nature employed there \cite[Theorem 4.4]{SwikirIfac} to transfer the additive form of simulation function to a max form (similar arguments can be found also in \cite[Remark 4.5]{SWIKIR2019}).   
Thus, our proposed results here outperform the ones in \cite{SwikirIfac} while providing more accurate overall abstractions. 

\begin{figure*}[t!]
	\centering
	\begin{subfigure}[b]{0.475\textwidth}
		\centering
		\includegraphics[width=\textwidth]{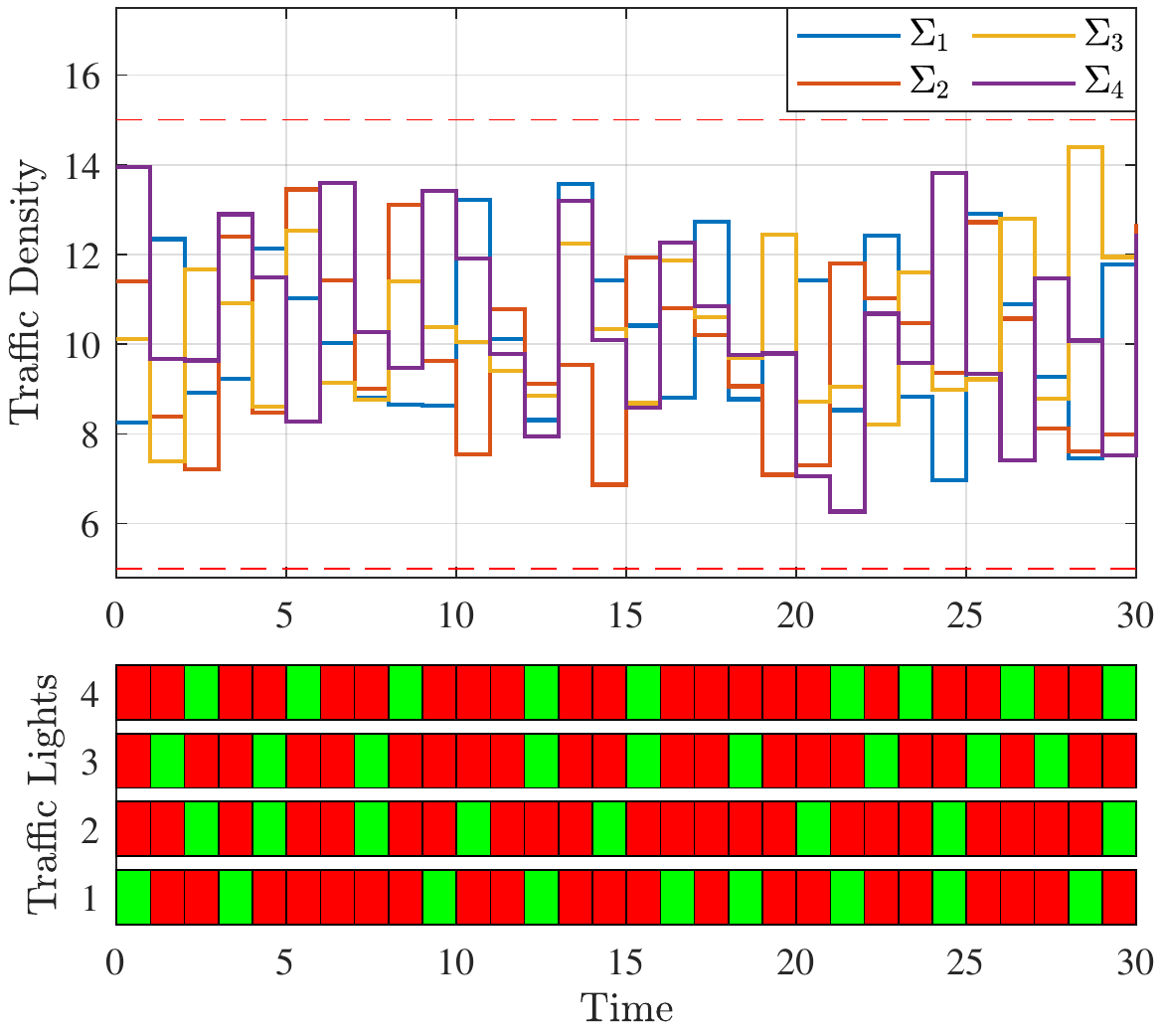}
		\caption{{\small Trajectories in subnetwork $\bar{G}_1$}}    
		\label{fig:trajsa}
	\end{subfigure}
	\hfill
	\begin{subfigure}[b]{0.475\textwidth}  
		\centering 
		\includegraphics[width=\textwidth]{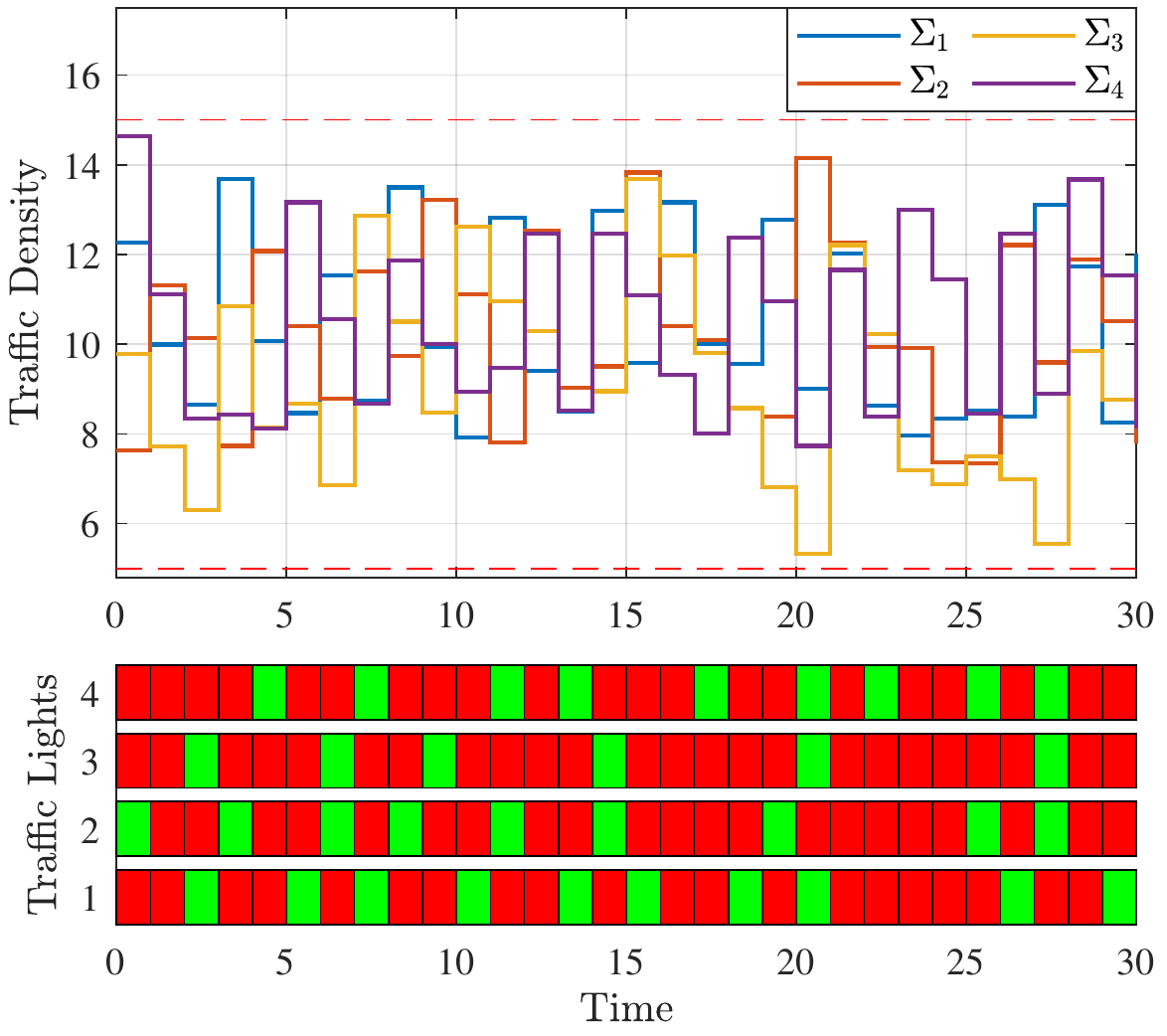}
		\caption{{\small Trajectories in subnetwork $\bar{G}_2$}}    
		\label{fig:trajsb}
	\end{subfigure}
	\vskip\baselineskip
	\begin{subfigure}[b]{0.475\textwidth}   
		\centering 
		\includegraphics[width=\textwidth]{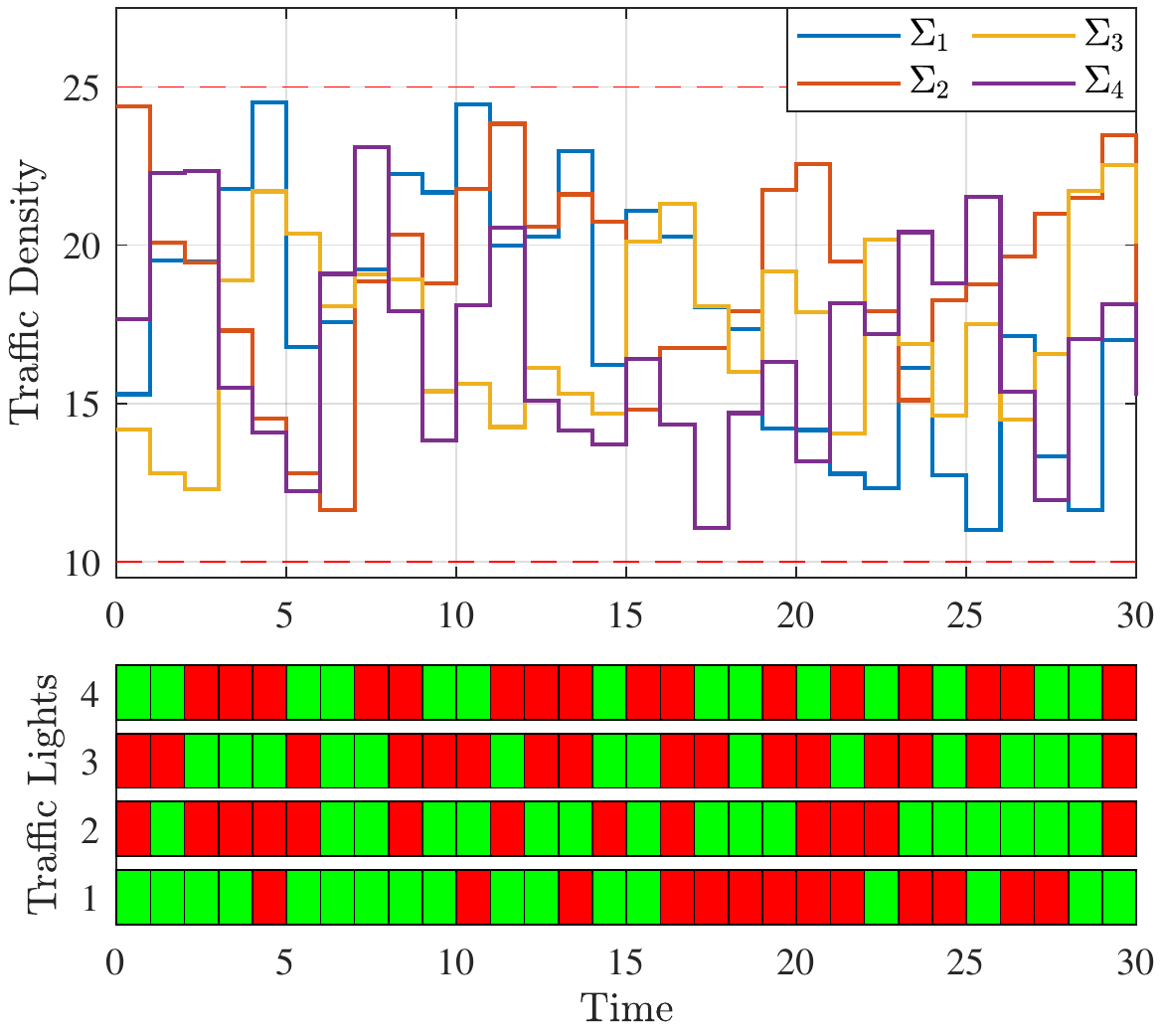}
		\caption{{\small Trajectories in subnetwork $\bar{G}_3$}}    
		\label{fig:trajsc}
	\end{subfigure}
	\hfill
	\begin{subfigure}[b]{0.475\textwidth}   
		\centering 
		\includegraphics[width=\textwidth]{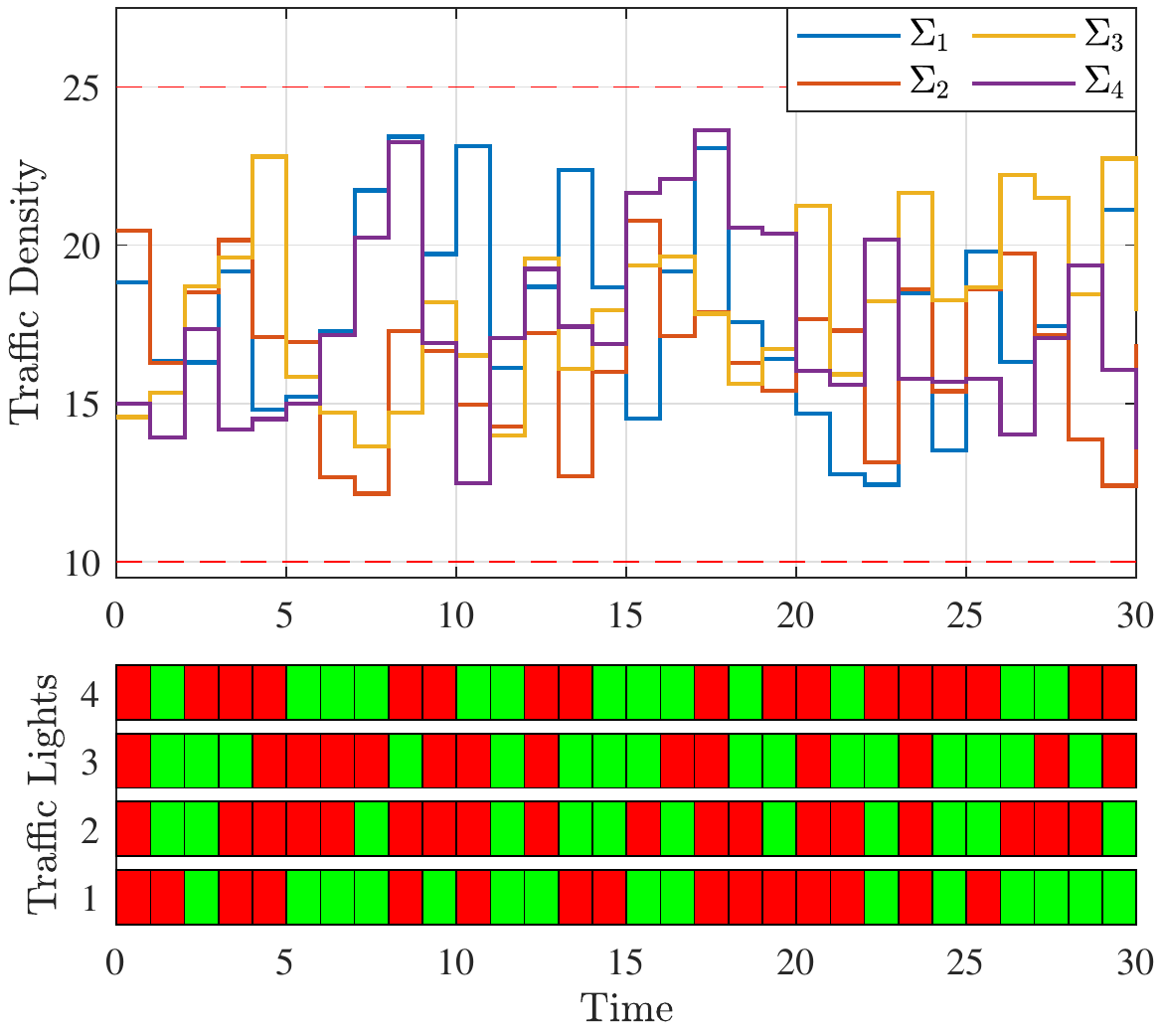}
		\caption{{\small Trajectories in subnetwork $\bar{G}_4$}}    
		\label{fig:trajsd}
	\end{subfigure}
	\caption{Simulation results: Trajectories of traffic density (upper subplots) and traffic lights (lower subplots) in sample cells from different subnetworks. The traffic density in each cell (subsystem $\Sigma_{i}$) is required to remain in desired safe region $S_i$ (indicated by the red dashed lines). The sets $S_i$ are given by $S_i = [5,15]$ in subnetworks $\bar{G}_1$ and $\bar{G}_2$, $S_i = [10,25]$ in subnetworks $\bar{G}_3$ and $\bar{G}_4$.} 
	\label{fig:trajs}
\end{figure*} 
\subsection{Compositional safety controller synthesis}
Now we synthesize a \emph{safety} controller for the infinite network via the constructed symbolic model such that the density of traffic in each cell is maintained in a desired safe region. Specifically, we aim at finding a control policy such that in subnetworks $\bar{G}_k$, $k \in \{1,2\}$, each subsystem $\Sigma_{i}$ satisfies safety specification $S_{i} = [5, 15]$ (vehicles per cell), and in subnetworks $\bar{G}_k$, $k \in \{3,4\}$, each subsystem $\Sigma_{i}$ satisfies safety specification $S_{i} = [10, 25]$ (vehicles per cell).    
Note that for the overall network, the overall safety specification  $S=\prod_{i\in\N}S_{i}$ is globally decomposable. By Proposition \ref{Theoremcontroler}, one can design local safety controllers for the subsystems separately with respect to local safety specifications, with the guarantee that the composed controller works as the overall safety controller for the overall infinite network.    
For each subsystem $\Sigma_i$, the idea is to design a local safety controller for its symbolic model $\hat{\Sigma}_i$, and then refine the controller back to the original subsystem by choosing $u_i = \hat u_i$. The control strategies are correct-by-construction, in the sense that the safety specification is guaranteed to be satisfied from any initial condition in the safe region.

Here, we employ the software tool \texttt{SCOTS} \cite{Rungger} to compositionally construct symbolic models and compute local safety controllers for subsystems $\Sigma_i$ with quantization parameters  $\eta^x_i = 0.1$ and $\eta^u_i = 0$, for each $i \in \N$.  
%All computations were conducted on a PC with Intel Core i7 3.4 GHz CPU. 
Computing symbolic models and synthesizing controllers for each subsystem took on average $0.006$s and $0.0004$s, respectively, on a PC with Intel Core i7 3.4 GHz CPU.
For each subnetwork, we show in Figure~\ref{fig:trajs} four sample state trajectories (upper plots of the sub-figures) and the corresponding input trajectories (lower plots of the sub-figures) of sample subsystems starting from random initial conditions.  
As can be seen in the figures, at each time step, the synthesized controllers react to the change in the density of the traffic in the corresponding cells by turning the traffic lights green/red. It can be observed that the density of the traffic using the synthesized controllers always remain in the desired safe regions.

\section{Conclusion}

In this paper, we proposed a methodology to compositionally
construct symbolic models for infinite networks.
To do this, we first introduced a notion of so-called alternating simulation functions that can be used to relate infinite networks.
A compositional approach was then proposed to construct symbolic models locally for concrete subsystems under incremental input-to-state stability property. 
By leveraging max-type small-gain type conditions, we provided an algorithm as a guideline for the design of local quantization parameters, such that the symbolic model of the infinite network can satisfy any given desired approximation accuracy.
A decentralized controller synthesis approach was presented to enforce safety properties on the overall infinite network. 
Finally, we applied our results on a road traffic network to verify the effectiveness of our compositionality results.

%\section{Bibliography styles}
%
%There are various bibliography styles available. You can select the style of your choice in the preamble of this document. These styles are Elsevier styles based on standard styles like Harvard and Vancouver. Please use Bib\TeX\ to generate your bibliography and include DOIs whenever available.
%
%Here are two sample references: \cite{Feynman1963118,Dirac1953888}.

\section*{Acknowledgment}
The authors would like to thank Abdalla Swikir for his fruitful discussions.

\bibliographystyle{IEEEtran}      
\bibliography{mybibfile} 

\begin{thebibliography}{10}
\providecommand{\url}[1]{#1}
\csname url@samestyle\endcsname
\providecommand{\newblock}{\relax}
\providecommand{\bibinfo}[2]{#2}
\providecommand{\BIBentrySTDinterwordspacing}{\spaceskip=0pt\relax}
\providecommand{\BIBentryALTinterwordstretchfactor}{4}
\providecommand{\BIBentryALTinterwordspacing}{\spaceskip=\fontdimen2\font plus
\BIBentryALTinterwordstretchfactor\fontdimen3\font minus
  \fontdimen4\font\relax}
\providecommand{\BIBforeignlanguage}[2]{{%
\expandafter\ifx\csname l@#1\endcsname\relax
\typeout{** WARNING: IEEEtran.bst: No hyphenation pattern has been}%
\typeout{** loaded for the language `#1'. Using the pattern for}%
\typeout{** the default language instead.}%
\else
\language=\csname l@#1\endcsname
\fi
#2}}
\providecommand{\BIBdecl}{\relax}
\BIBdecl

\bibitem{infplatoon}
M.~R. {Jovanovic} and B.~{Bamieh}, ``On the ill-posedness of certain vehicular
  platoon control problems,'' \emph{IEEE Transactions on Automatic Control},
  vol.~50, no.~9, pp. 1307--1321, 2005.

\bibitem{Bamiehtac}
B.~{Bamieh}, F.~{Paganini}, and M.~A. {Dahleh}, ``Distributed control of
  spatially invariant systems,'' \emph{IEEE Transactions on Automatic Control},
  vol.~47, no.~7, pp. 1091--1107, 2002.

\bibitem{kawan2019lyapunov}
C.~Kawan, A.~Mironchenko, A.~Swikir, N.~Noroozi, and M.~Zamani, ``A
  {L}yapunov-based small-gain theorem for infinite networks,'' \emph{IEEE
  Transactions on Automatic Control, in press}, 2021.

\bibitem{mironchenko2020nonlinear}
A.~Mironchenko, C.~Kawan, and J.~Gl{\"u}ck, ``Nonlinear small-gain theorems for
  input-to-state stability of infinite interconnections,'' \emph{arXiv preprint
  arXiv:2007.05705}, 2020.

\bibitem{dashkovskiy2019stability}
S.~Dashkovskiy, A.~Mironchenko, J.~Schmid, and F.~Wirth, ``Stability of
  infinitely many interconnected systems,'' in \emph{Proceedings of 11th IFAC
  Symposium on Nonlinear Control Systems.}\hskip 1em plus 0.5em minus
  0.4em\relax Elsevier, 2019, pp. 937--942.

\bibitem{navidcdc}
N.~{Noroozi}, A.~{Mironchenko}, and F.~R. {Wirth}, ``A relaxed small-gain
  theorem for discrete-time infinite networks,'' in \emph{59th IEEE Conference
  on Decision and Control}, 2020, pp. 3102--3107.

\bibitem{mironchenko2020input}
A.~Mironchenko and C.~Prieur, ``Input-to-state stability of
  infinite-dimensional systems: Recent results and open questions,'' \emph{SIAM
  Review}, vol.~62, no.~3, pp. 529--614, 2020.

\bibitem{dashkovskiy2020stability}
S.~Dashkovskiy and S.~Pavlichkov, ``Stability conditions for infinite networks
  of nonlinear systems and their application for stabilization,''
  \emph{Automatica}, vol. 112, p. 108643, 2020.

\bibitem{SwikirIfac}
A.~Swikir, N.~Noroozi, and M.~Zamani, ``Compositional synthesis of symbolic
  models for infinite networks,'' in \emph{21st IFAC World Congress}, July
  2020.

\bibitem{Tabuada.2009}
P.~Tabuada, \emph{Verification and control of hybrid systems: a symbolic
  approach}.\hskip 1em plus 0.5em minus 0.4em\relax Boston, MA: Springer, 2009.

\bibitem{Pola.2009}
G.~Pola and P.~Tabuada, ``Symbolic models for nonlinear control systems:
  Alternating approximate bisimulations,'' \emph{{SIAM} Journal on Control and
  Optimization}, vol.~48, no.~2, pp. 719--733, 2009.

\bibitem{Pola.2008}
G.~Pola, A.~Girard, and P.~Tabuada, ``Approximately bisimilar symbolic models
  for nonlinear control systems,'' \emph{Automatica}, vol.~44, no.~10, pp.
  2508--2516, 2008.

\bibitem{baier2008principles}
C.~Baier and J.-P. Katoen, \emph{Principles of model checking}.\hskip 1em plus
  0.5em minus 0.4em\relax MIT press, 2008.

\bibitem{tazaki2008bisimilar}
Y.~Tazaki and J.-i. Imura, ``Bisimilar finite abstractions of interconnected
  systems,'' in \emph{International Workshop on Hybrid Systems: Computation and
  Control}.\hskip 1em plus 0.5em minus 0.4em\relax Springer, 2008, pp.
  514--527.

\bibitem{Pola.2016}
G.~Pola, P.~Pepe, and M.~D.~D. Benedetto, ``Symbolic models for networks of
  control systems,'' \emph{{IEEE} Transactions on Automatic Control}, vol.~61,
  no.~11, pp. 3663--3668, Nov. 2016.

\bibitem{Majid18}
M.~{Zamani} and M.~{Arcak}, ``Compositional abstraction for networks of control
  systems: A dissipativity approach,'' \emph{IEEE Transactions on Control of
  Network Systems}, vol.~5, no.~3, pp. 1003--1015, 2018.

\bibitem{SWIKIR2019}
A.~Swikir and M.~Zamani, ``Compositional synthesis of finite abstractions for
  networks of systems: A small-gain approach,'' \emph{Automatica}, vol. 107,
  no.~11, pp. 551 -- 561, 2019.

\bibitem{mallik2018compositional}
K.~Mallik, A.-K. Schmuck, S.~Soudjani, and R.~Majumdar, ``Compositional
  synthesis of finite-state abstractions,'' \emph{IEEE Transactions on
  Automatic Control}, vol.~64, no.~6, pp. 2629--2636, 2018.

\bibitem{meyer}
P.~J. Meyer, A.~Girard, and E.~Witrant, ``Compositional abstraction and safety
  synthesis using overlapping symbolic models,'' \emph{IEEE Transactions on
  Automatic Control}, vol.~63, no.~6, pp. 1835--1841, 2017.

\bibitem{kim2018constructing}
E.~S. Kim, M.~Arcak, and M.~Zamani, ``Constructing control system abstractions
  from modular components,'' in \emph{Proceedings of the 21st International
  Conference on Hybrid Systems: Computation and Control}, 2018, pp. 137--146.

\bibitem{pola2010symbolic}
G.~Pola, P.~Pepe, M.~D. Di~Benedetto, and P.~Tabuada, ``Symbolic models for
  nonlinear time-delay systems using approximate bisimulations,'' \emph{Systems
  \& Control Letters}, vol.~59, no.~6, pp. 365--373, 2010.

\bibitem{girard2014approximately}
A.~Girard, ``Approximately bisimilar abstractions of incrementally stable
  finite or infinite dimensional systems,'' in \emph{53rd IEEE Conference on
  Decision and Control}, 2014, pp. 824--829.

\bibitem{angeli2002lyapunov}
D.~Angeli, ``A {L}yapunov approach to incremental stability properties,''
  \emph{IEEE Transactions on Automatic Control}, vol.~47, no.~3, pp. 410--421,
  2002.

\bibitem{zamani2014symbolic}
M.~Zamani, P.~Mohajerin~Esfahani, R.~Majumdar, A.~Abate, and J.~Lygeros,
  ``Symbolic control of stochastic systems via approximately bisimilar finite
  abstractions,'' \emph{IEEE Transactions on Automatic Control}, vol.~59,
  no.~12, pp. 3135--3150, 2014.

\bibitem{henzinger1998you}
T.~A. Henzinger, S.~Qadeer, and S.~K. Rajamani, ``You assume, we guarantee:
  Methodology and case studies,'' in \emph{Computer Aided Verification}, 1998,
  pp. 440--451.

\bibitem{liu2019compositional}
S.~Liu and M.~Zamani, ``Compositional synthesis of almost maximally permissible
  safety controllers,'' in \emph{American Control Conference}, 2019, pp.
  1678--1683.

\bibitem{canudas2012graph}
C.~Canudas-de Wit, L.~L. Ojeda, and A.~Y. Kibangou, ``Graph constrained-ctm
  observer design for the grenoble south ring,'' \emph{IFAC Proceedings
  Volumes}, vol.~45, no.~24, pp. 197--202, 2012.

\bibitem{Rungger}
M.~Rungger and M.~Zamani, ``S{C}{O}{T}{S}: A tool for the synthesis of symbolic
  controllers,'' in \emph{Proceedings of the 19th International Conference on
  Hybrid Systems: Computation and Control}.\hskip 1em plus 0.5em minus
  0.4em\relax {ACM}, Apr. 2016.

\end{thebibliography}

\end{document}